\documentclass{amsart}
\usepackage[utf8]{inputenc}

\usepackage{amsmath,amssymb,mathrsfs,amsthm}
\allowdisplaybreaks[4] 
\usepackage{mathtools} 
\usepackage{booktabs}
\usepackage[usenames]{color}
\usepackage{color}  
\usepackage{comment}
\usepackage{graphicx} 
\usepackage{tabularx}
\usepackage{supertabular} 
\usepackage{enumerate} 
\usepackage{url}
\usepackage{tikz} 
\usepackage{pgfplots}  
\usepackage{bbm}
\usetikzlibrary{intersections, calc} 
\usetikzlibrary{patterns}  
\usetikzlibrary{arrows,decorations.markings} 
\usetikzlibrary {lindenmayersystems}  
\usepackage{diffcoeff}
\usepackage{mathtools}
\mathtoolsset{showonlyrefs=true}
\usepackage{fontspec}
\usepackage[
  warnings-off={mathtools-colon,mathtools-overbracket}
]{unicode-math}

\newcommand{\xhdr}[1]{\noindent \textup{{\noindent\bfseries #1}:}}

\pgfplotsset{compat=1.18}

\newtheorem{theorem}{Theorem}[section]
\newtheorem{mainthm}{Theorem} 

\newtheorem{lemma}[theorem]{Lemma}

\newtheorem{remark}[theorem]{Remark}
\newtheorem{definition}[theorem]{Definition}
\newtheorem{corollary}[theorem]{Corollary}
\newtheorem{example}[theorem]{Example}

\title[A Leibniz rule of distributional pairing and hyperforce sum rule] 
{A Leibniz rule of distributional pairing and hyperforce sum rule} 

\author[T. MARUYAMA]{Takashi MARUYAMA}
\address{NEC Secure System Platform Research Laboratories, 1753 Shimonumabe, Nakahara-ku, Kawasaki, Kanagawa, Japan}
\email{49takashi@nec.com}

\author[T. SETO]{Tatsuki SETO}
\address{General Education and Research Center, Meiji Pharmaceutical University, 
2-522-1 Noshio, Kiyose-shi, Tokyo, Japan}
\email{tatsukis@my-pharm.ac.jp}

\author[V. Zaverkin]{Viktor Zaverkin}
\address{NEC Laboratories Europe, Kurfürsten-Anlage 36, 69115 Heidelberg, Germany}
\email{viktor.zaverkin@neclab.eu}

\author[H. Christiansen]{Henrik Christiansen}
\address{NEC Laboratories Europe, Kurfürsten-Anlage 36, 69115 Heidelberg, Germany}
\email{henrik.christiansen@neclab.eu}


\begin{document}

\begin{abstract} 
We reformulate and generalize the equilibrium hyperforce sum rule, a generalization of the Bogoliubov–Born–Green–Kirkwood–Yvon (BBGKY) hierarchy, by employing the Schwartz space and its dual. We show that the hyperforce sum rule for the Euclidean space and the equilibrium BBGKY hierarchy at arbitrary level are derived through the Leibniz rule of the derivative for the pairing of tempered distributions and Schwartz functions. We also apply the Leibniz rule to obtain the hyperforce sum rule for systems with periodic boundary conditions.
\end{abstract}

\maketitle

\section*{Introduction}  
The purpose of this paper is to give a reformulation and generalization of the equilibrium hyperforce sum rule obtained in references \cite{PhysRevLett.133.217101, 168}, by using the Schwartz space and its dual. 
The hyperforce sum rule is a generalization of the fully-reduced Yvon-Born-Green equation (YBG) or Bogoliubov–Born–Green–Kirkwood–Yvon (BBGKY) hierarchy \cite{hansen2013theory}, one representative class of statistical mechanical sum formulae \cite{henderson2021fundamentals}. Our distributional formulation gives a natural generalization of the original hyperforce sum rule and subsumes the equilibrium BBGKY hierarchy an any level. It also naturally applies to systems with periodic boundary conditions. 

The hyperforce sum rule \cite{PhysRevLett.133.217101, 168} is obtained by employing Noether's theorem \cite{Noether} of variational invariance for the thermal average of observables over the Boltzmann distribution. The prototypical application of the invariant theorem has been studied in a range of scenarios in statistical mechanics \cite{Hermann_2024, Hermann2021-fp, Hermann_2022_force_balance, Hermann2022-os, whyNoether, Robitschko2024-zj, PhysRevLett.130.268203, PhysRevE.107.034109, PhysRevE.106.014115}. 
A key tool that enables the derivation of the hyperforce sum rule is the functional derivative of the thermal average, which is invariant to the action of canonical transformations that preserve the Hamiltonian. 
Specifically, the thermal average is regarded as a constant functional on a space of diffeomorphisms lifted from the configuration space, which crucially relies on the invariance of the thermal average with respect to the action of the canonical transformations over configuration space. 
Since the functional is constant on the diffeomorphisms, the functional derivative of the thermal average vanishes on this space (of diffeomorphisms). This vanishing property is called the hyperforce sum rule. References \cite{PhysRevLett.133.217101, 168} also show that, when the observable is set to be a (non-trivial) constant function, the hyperforce sum rule can restore the fully-reduced BBGKY hierarchy. This pathway contrasts to the standard derivation of the BBGKY hierarchy as a corollary of the Liouville equation \cite{hansen2013theory}. We will review the derivation of the BBGKY hierarchy and hyperforce sum rule in Section \ref{sec:prelim}.

We introduce some notations before giving the overview of the paper's main result. For the sake of brevity, we leave out their detailed definition in this section.
Let $N$ be the number of particles in the Euclidean space $\symbb{R}^{d}$. 
Let $\symfrak{X}_{0}(\symbb{R}^{d})$ be the set of vector fields on $\symbb{R}^{d}$  with compact support. 
Set $\symfrak{X}_{\mathrm{Id},0}(\symbb{R}^{d})$ to be the subspace of $\symfrak{X}_{0}(\symbb{R}^{d})$ 
such that $\mathrm{Id} + \epsilon: \symbb{R}^{d} \rightarrow \symbb{R}^{d}$ is a diffeomorphism, where $\epsilon \in \symfrak{X}_{0}(\symbb{R}^{d})$. 
Let $\symscr{S}(\symbb{R}^{dN \times 2})$ be the set of Schwartz functions (i.e., roughly, rapidly decaying smooth functions) and $\symscr{S}'(\symbb{R}^{dN \times 2})$ its dual space, that is the set of linear functionals, which are called tempered distributions. In other words, by definition, when both $\phi \in \symscr{S}(\symbb{R}^{dN \times 2})$ and $u \in \symscr{S}'(\symbb{R}^{dN \times 2})$ are given, we obtain a scalar $u(\phi) \in \symbb{C}$.  
Conventionally, $u(\phi)$ is written as the distributional pairing notation $\langle u , \phi \rangle$. 
It is well-known that continuous functions on $\symbb{R}^{dN \times 2}$ give rise to important examples of tempered distributions by assigning the integral: 
for a ``good'' function $f$ on $\symbb{R}^{dN \times 2}$, we obtain a tempered distribution $u_{f}$ by 
\begin{equation}
\langle u_{f} , \phi \rangle 
= \int_{\symbb{R}^{dN \times 2}} f\phi \, d\symbfit{r}_{1} \cdots d\symbfit{r}_{N} d\symbfit{p}_{1} \cdots d\symbfit{p}_{N}  
\end{equation}
for $\phi \in \symscr{S}(\symbb{R}^{dN \times 2})$. 
See also Appendix \ref{app:distribution} for the details of Schwartz functions and theory of tempered distributions.

The starting point of our work is an observation that the thermal average is an instance of the distributional pairing $\langle u , \phi \rangle$. 
Let $\hat{A}$ be the observable and $e^{-\beta H}$ the Boltzmann distribution\footnote{In physics literature, the Boltzmann distribution is sometimes considered as an exponentially decreasing function defined on the spectrum. In the present paper, we instead always consider that the Boltzmann distribution is a Schwartz function defined on the phase space $\symbb{R}^{dN \times 2}$.}. 
Then the thermal average of $\hat{A}$ is defined by 
\begin{equation}
\langle \hat{A} \rangle = \dfrac{1}{Z}\int_{\symbb{R}^{dN \times 2}} \hat{A} e^{-\beta H} \, d\symbfit{r}_{1} \cdots d\symbfit{r}_{N} d\symbfit{p}_{1} \cdots d\symbfit{p}_{N}, 
\end{equation}
where $Z$ is the partition function (or the normalizing constant) of $e^{-\beta H}$.
By the definition of the pairing notation, we see the following equation holds immediately 
\begin{equation}
\langle \hat{A} \rangle 
= \langle u_{\hat{A}} , e^{-\beta H}/Z \rangle . 
\end{equation} 

Similarly, the invariance of the thermal average by canonical transformations may be also understood through the distributional pairing. Indeed, the invariance of the thermal average translates to the ``pullback'' of tempered distributions. First, we consider the following diffeomorphism between the phase space $\symbb{R}^{dN \times 2}$ induced by $\epsilon \in \mathfrak{X}_{\operatorname{Id}, 0}(\mathbb{R}^{d})$: 
\begin{align}
\epsilon_{\sharp}&(\symbfit{r}_{1}, \dots , \symbfit{r}_{N}, \symbfit{p}_{1}, \dots , \symbfit{p}_{N}) \\ 
&= 
((\mathrm{Id} + \epsilon)(\symbfit{r}_{1}), \dots , (\mathrm{Id} + \epsilon)(\symbfit{r}_{N}), 
(\symbb{1} + \nabla\epsilon(\symbfit{r}_{1}))^{-1}\symbfit{p}_{1}, \dots , (\symbb{1} + \nabla\epsilon(\symbfit{r}_{N}))^{-1}\symbfit{p}_{N}). 
\end{align} 
For a tempered distribution $u \in \symscr{S}'(\symbb{R}^{dN \times 2})$, we define the pullback $\epsilon_{\sharp}^{\ast}u$ of $u$ by $\epsilon_{\sharp}$ as  
\begin{equation}
\hspace{15mm}\langle \epsilon_{\sharp}^{\ast}u , \phi \rangle 
= \langle u , \phi \circ \epsilon_{\sharp}^{-1} \rangle, \quad\phi \in \symscr{S}(\symbb{R}^{dN \times 2}).
\end{equation} 
By the definition of the pullback, the functional 
\[
\mathfrak{X}_{\operatorname{Id},0}(\symbb R^{d}) \rightarrow \symbb{C}, \quad 
\epsilon \mapsto \langle \epsilon_{\sharp}^{*}u , \phi \circ \epsilon_{\sharp} \rangle 
\]
is constant on $\mathfrak{X}_{\operatorname{Id},0}(\symbb R^{d})$, which is a distributional analogue of the invariance of the original thermal average by canonical transformations. As a result, its directional derivative, given as 
\[
F[u, \phi]: \mathfrak{X}_{0}(\symbb R^{d}) \rightarrow \symbb{C}, 
\quad 
\epsilon \mapsto \left. \diff{}{t} \right|_{t=0} \left\langle (t\epsilon)_{\sharp}^{*}u , \phi \circ (t\epsilon_{\sharp})\right\rangle, 
\] 
along any $\epsilon \in \mathfrak{X}_{0}(\mathbb{R}^{d})$ turns out to vanish. We will refer to this vanishing property as \textit{the equilibrium distributional hyperforce sum rule} for $u \in \symscr{S}^{\prime}(\symbb{R}^{dN\times2})$ and $\phi \in \symscr{S}(\mathbb{R}^{dN \times 2})$.

After introducing the necessary preliminaries, we show that there is another way to derive the distributional hyperforce sum rule. This route, which naively may look like a futile effort, alluringly makes it possible to derive a prototypical expression of the hyperforce sum rule. A tool that makes this derivation possible is the Leibniz rule of the pairing of tempered distributions, as we show as Theorem \ref{thm:product_rule} in Appendix \ref{app:distribution}, and rapidly decreasing functions:
\begin{equation}
\left. \diff{}{t} \right|_{t=0} \left\langle u_{t} , \phi_{t} \right\rangle 
= \left\langle \left. \diff{u_{t}}{t} \right|_{t=0} , \phi \right\rangle + \left\langle u, \left. \diff{\phi_{t}}{t} \right|_{t=0}\right\rangle,  
\end{equation} 
where we denote $u_{t} = (t\epsilon)_{\sharp}^{*}u$ and $\phi_{t} = \phi \circ (t\epsilon)_{\sharp}$. 
Applying the Leibniz rule to $F[u,\phi]$ yields two mutually vanishing terms, which serve as the main components of the distributional hyperforce sum rule. The two-term expression further yields a shifting operator $D_{i}$, a conceptual counterpart of the infinitesimal shifting operator $\sigma_{i}(\symbfit{r})$ introduced by \cite{PhysRevLett.133.217101, 168}.

The hitherto discussion can generalize to the $n$-body reduced thermal average, where the case $n=0$ corresponds to the discussion above.
Our main theorem in the general setting is as follows. 

\begin{mainthm}[Theorem \ref{thm:hyperforce_sum_rewrite} and Corollary \ref{cor:hyperforcesum}] \label{thm:first_main_result}
Let $N \in \symbb{N}$ and $0 \leq n \leq N-1$. For a tempered distribution $u \in \symscr{S}^{\prime}(\symbb{R}^{d(N-n)\times2})$ and a Schwartz function $\phi \in \symscr{S}(\symbb{R}^{dN\times2})$, we define the \textup{localized hyperforce of $u$ and $\phi$ centered at $i$} to be a continuous-function-valued map 
$F^{(n)}_{i}[u, \phi]: \mathfrak{X}_{0}(\mathbb{R}^{d}) \rightarrow C\left(\symbb{R}^{dn\times2}\right)$ as: 
\begin{align}
F_{i}^{(n)}[u,\phi][\epsilon] = \left\langle K_{n}[D_{i}(\epsilon)u] , \phi \right\rangle + \left\langle K_{n}[u] , D_{i}(\epsilon)\phi \right\rangle. 
\end{align}
Then, the sum of the localized hyperforce coincides with the equilibrium distributional hyperforce sum $F^{(n)}[u, \phi]$ and vanishes, i.e., 
\begin{align}
    F^{(n)}[u, \phi] = \sum_{i=n+1}^{N}F_{i}^{(n)}[u, \phi] = 0.
\end{align}
\end{mainthm} 
\noindent

Our distributional interpretation of the hyperforce sum rule offers a couple of advantages: First, our formulation naturally subsumes the equilibrium BBGKY hierarchy at any level, as well as the original hyperforce sum rule. This is made possible since our formulation is based on Banach space-valued distributions, which allows to include the space of continuous functions $C(\symbb{R}^{dn \times 2})$, giving a holistic view to the BBGKY hierarchy and the hyperforce sum rule. Furthermore, our formulation is naturally applicable to systems with periodic boundary conditions, whose results are shown in Section \ref{sec:compact}. Table \ref{tab:comparison} summarizes the correspondence between the notions introduced in our work and those in \cite{PhysRevLett.133.217101, 168}.

\begin{table}[h!]
\centering
\caption{Correspondence of notions between our work and the relevant works.}
\label{tab:comparison}

\resizebox{0.7\textwidth}{!}{
\begin{tabular}{c|c} 
\toprule
\textbf{Our work} & \cite{PhysRevLett.133.217101, 168} \\ 
\midrule
Schwartz function & Boltzmann factor \\ 
Tempered distribution & Observable \\ 
Lift of diffeomorphism & Canonical transformation\\
Leibniz rule  & Functional derivative\\
Distributional hyperforce sum rule & Hyperforce sum rule\\
\bottomrule
\end{tabular}
}
\end{table}

After deriving Theorem \ref{thm:first_main_result}, we show that it yields the BBGKY hierarchies and the original hyperforce sum rule. We begin the derivation by clarifying a condition on Hamiltonian functions that is necessary for the corresponding Boltzmann distributions to belong to the Schwartz space. One typical class of such Hamiltonian functions is the many-body Hamiltonian 
\[
H(\symbfit{r}^N, \symbfit{p}^N) = \sum_{i=1}^{N}\frac{\symbfit{p}_{i}^{\operatorname{T}} \symbfit{p}_{i}}{2m_{i}} + u_{N}(\symbfit{r}^N) + \sum_{i=1}^{N}u^\textup{ext}(\symbfit{r}_i),
\]
where $m_{i}$ is the mass of the particle $i$.  Being rapidly decreasing is essential to allow such class of Hamiltonians to include physically important class of potentials that involve both harmonic potential and repulsive potential functions, such as the Lennard-Jones potential \cite{lennard_jones_second, lennard_jones_first} and Coulomb potential \cite{halliday2013fundamentals}.

We show that when the tempered distribution is realized as the integration of a function, the localized hyperforce $F_{i}^{(n)}$ admits a functional representation of the integral type for compactly supported vector fields $\epsilon \in \mathfrak{X}_{0}(\mathbb{R}^{d})$. 

\begin{mainthm}[Lemma \ref{cor:general_derivation_hyperforce_sum_rule} and Theorem \ref{thm:generalized_BBGKY}]
Let $\beta = \frac{1}{k_\mathrm{B}T}$ be the inverse temperature, where $k_\mathrm{B}$ denotes the Boltzmann constant and $T$ the absolute temperature. Assume $u_{N}$ and $u^{\textup{ext}}$ are chosen such that $e^{-\beta H(\symbfit{r}^N, \symbfit{p}^N)} \in \symscr{S}(\symbb{R}^{dN \times 2})$ and $f$ is a tempered $C^{1}$ function on $\symbb{R}^{dN \times 2}$. 
Then, there exists a $C(\symbb{R}^{dn \times 2})^{d}$-valued functions $G_{i}^{(n)}[f,H]$ on $\symbb{R}^{d}$ for all $i \in \{n+1,\dots, N\}$ such that 
\begin{equation}
F_{i}^{(n)}[u_{f_{n}}, e^{-\beta H(\symbfit{r}^N, \symbfit{p}^N)}][\epsilon] 
= \int_{\symbb{R}^{d}}\epsilon(\symbfit{r}_{i})^{\operatorname{T}} G_{i}^{(n)}[f, H](\symbfit{r}_{i}) d\symbfit{r}_{i}. 
\end{equation} 
Furthermore, $G_{i}^{(n)}[f,H]$ vanishes on $\symbb{R}^{d}$.
\end{mainthm}
 
Consequently, the vanishing result of $G_{i}^{(n)}[f,H]$ turns out to include the BBGKY hierarchy and the hyperforce sum rule as its instances, which we show in Corollary \ref{cor:derivation_hyperforce_sum_rule} and Corollary \ref{cor:generalized_168}. 

\begin{mainthm} [Corollary \ref{cor:derivation_hyperforce_sum_rule} and Corollary \ref{cor:generalized_168}]
We assume $\displaystyle u_{N}(\symbfit{r}^{N}) = \sum_{1 \leq i < j \leq N} u(\symbfit{r}_i, \symbfit{r}_j)$ with $u$ being symmetric. 
Then, for any $n = 1,2,\dots , N-1$ and $k = 1,2,\dots, n$, we have 
\begin{align}
&\left( \frac{\symbfit{p}_{k}}{m_{k}} \cdot \nabla_{\symbfit{r}_{k}}  - \left(\nabla_{\symbfit{r}_{k}} u^{\text{ext}} (\symbfit{r}_{k}) + \sum_{\substack{j=1 \\ j \neq k}}^{n} \left(\nabla_{\symbfit{r}_{k}} u (\symbfit{r}_{k}, \symbfit{r}_{j}) \right) \right) \cdot \nabla_{\symbfit{p}_{k}} \right) \phi^{[n]} \\ 
&= \int_{\symbb{R}^{2d}}  \left(\nabla_{\symbfit{r}_{k}} u(\symbfit{r}_{k}, \symbfit{r}_{n+1}) \right) \cdot \left(\nabla_{\symbfit{p}_{k}} \phi^{[n+1]}\right) d\symbfit{r}_{n+1} d\symbfit{p}_{n+1}. 
\end{align}
Here, $\phi^{[n]}$ is the reduced phase-space distribution function for $\phi$ defined in Section \ref{sec:bbgky}. 
Furthermore, the vanishing of a map $G^{(n)}[f, H]: \symbb{R}^{d} \rightarrow C(\symbb{R}^{dn \times 2})^{d}$ defined as 
\[
G^{(n)}[f, H](\symbfit{r}) = \sum_{i=1}^{N} G^{(n)}_{i}[f, H](\symbfit{r})
\]
coincides with the original hyperforce sum rule 
when $n=0$. 
\end{mainthm}

\section*{Acknowledgements} 
T.M. is deeply grateful to Federico Errica and Francesco Alesiani for useful suggestions and continuous encouragement. T.M. gratefully acknowledges the support of NEC Laboratories Europe during his expatriation at the laboratory.

\section{Preliminary and assumption}
\label{sec:prelim}
We recall in this section one standard derivation of Bogolyubov–Born–Green–Kirkwood–Yvon (BBGKY) hierarchy, based on \cite{hansen2013theory}, and its generalization referred to as hyperforce sum rule derived in \cite{168}. This section generally follows a conventional physics terminology and symbols. Throughout this paper, we consider that any system consists of finitely many multiple particles and is in equilibrium. Therefore, any variables introduced in the rest of this paper are time-independent.
\subsection{Bogolyubov–Born–Green–Kirkwood–Yvon (BBGKY) hierarchy} \label{sec:bbgky}
We first consider the BBGKY hierarchy in a general setting including non-equilibrium systems. We consider the phase space $\mathbb R^{dN \times 2}$ on $\symbb R^{dN}$ with linear coordinates $(\symbfit{r}_{1}, \cdots, \symbfit{r}_{N}, \symbfit{p}_{1}, \cdots, \symbfit{p}_{N})$ representing spatial coordinates and momenta of $N$ bodies. We also assume that the bodies are confined in a compact space in $\symbb{R}^d$ ($d \geq 1$).

Let $H$ be a Hamiltonian on the phase space $\symbb{R}^{dN \times 2}$. We assume that $H$ is separable for a general potential $u_N\left(\symbfit{r}^N\right)$, and write $H$ in general form as 
\[H(\symbfit{r}^{N}, \symbfit{p}^{N}) = K_{N}(\symbfit{p}^N) + u_{N}(\symbfit{r}^{N}) + \sum_{k=1}^{N} u^\text{ext} (\symbfit{r}_{k}),\]
where $\displaystyle K_{N}(\symbfit{p}^N) = \sum_{k=1}^{N}\frac{|\symbfit{p}_{k}|^{2}}{2m_{k}}$ with a body mass $m_{k}$, $u_{N}$ is the interatomic potential energy and $u^\text{ext}$ is the external 1-body potential energy. Let $\phi^{[N]}(\symbfit{r}^N, \symbfit{p}^N, t)$ be a phase-space probability density. Suppose that $\phi^{[N]}(\symbfit{r}^N, \symbfit{p}^N, t)$ satisfies the Liouville equation
\[\frac{\partial \phi^{[N]}}{\partial t} + \sum_{k=1}^{N} \left( \left(\nabla_{\symbfit{r}_{k}} \phi^{[N]}\right) \cdot \dot{\symbfit{r}}_{k} + \left(\nabla_{\symbfit{p}_{k}} \phi^{[N]}\right) \cdot \dot{\symbfit{p}}_{k} \right) = 0.\]
Here, $\nabla_{\symbfit{r}} = (\frac{\partial}{\partial r_{1}}, \dots, \frac{\partial}{\partial r_{d}})^{\operatorname{T}}$. This equation is rewritten, by using 
\[\nabla_{\symbfit{r}_{k}} H = \nabla_{\symbfit{r}_{k}} u_{N}(\symbfit{r}^{N}) + \nabla_{\symbfit{r}} u^{\text{ext}}(\symbfit{r}_{k}),\]
as 
\[\left( \frac{\partial}{\partial t} + \sum_{k=1}^{N}\frac{\symbfit{p}_{k}}{m_{k}} \cdot \nabla_{\symbfit{r}_{k}} - \sum_{k=1}^{N} \left(\nabla_{\symbfit{r}_{k}} u^{\text{ext}} \right) \cdot \nabla_{\symbfit{p}_{k}}\right) \phi^{[N]} = \sum_{k=1}^{N} \left(\nabla_{\symbfit{r}_{k}} u_{N} \right) \left(\nabla_{\symbfit{p}_{k}} \phi^{[N]} \right).\]
We introduce a $n$-body reduced phase-space distribution function
\[
\phi^{[n]}(\symbfit{r}^n, \symbfit{p}^n;t) = \frac{N!}{(N-n)!}\iint \phi^{[N]}d\symbfit{r}^{N}_{n} d\symbfit{p}^{N}_{n},
\]
in which $d\symbfit{r}^{N}_{n} = d\symbfit{r}_{n+1} \dots d\symbfit{r}_N$ and $d\symbfit{p}^{N}_{n} = d\symbfit{p}_{n+1} \dots d\symbfit{p}_N$. 
Then, applying the expression to the above Liouville equation yields
\begin{equation} 
\begin{aligned}
&\phantom{=}\left( \frac{\partial}{\partial t} + \sum_{k=1}^{n}\frac{\symbfit{p}_k}{m_{k}} \cdot \nabla_{\symbfit{r}_{k}}  - \sum_{k=1}^{n} \left(\nabla_{\symbfit{r}_{k}} u^{\text{ext}}\right) \cdot \nabla_{\symbfit{p}_{k}} \right) \phi^{[n]} \\ 
&= \sum_{k=1}^{n} \sum_{j=1}^{n} \left(\nabla_{\symbfit{r}_{k}} u_{N} (\symbfit{r}_k, \symbfit{r}_j) \right) \cdot \left(\nabla_{\symbfit{p}_{k}}  \phi^{[n]} \right) \\
&\hspace{10mm} + \frac{N!}{(N-n)!} \sum_{k=1}^{n} \sum_{j=n+1}^{N}\iint  \left(\nabla_{\symbfit{r}_{k}} u_{N}(\symbfit{r}_{k}, \symbfit{r}_{j}) \right) \cdot \left(\nabla_{\symbfit{p}_{k}} \phi^{[N]}\right) d\symbfit{r}^{N}_{n} d\symbfit{p}^{N}_{n}.
\end{aligned} 
\label{eq:Liouville_f^{[n]}}
\end{equation} 
Here, we assume that $u_{N}(\symbfit{r}^N)$ is a pair-wise symmetric potential energy: 
\[
u_{N}(\symbfit{r}^N)= \sum_{1 \leq k < j \leq N} u(\symbfit{r}_{k}, \symbfit{r}_{j}).
\]
Noting that $H$ is symmetric, the equation \eqref{eq:Liouville_f^{[n]}} is reduced to an expression known as BBGKY hierarchy:

\begin{align}
&\phantom{=}\left( \frac{\partial}{\partial t} + \sum_{k=1}^{n}\frac{\symbfit{p}_k}{m_{k}} \cdot \nabla_{\symbfit{r}_{k}}  - \sum_{k=1}^{n} \left(\nabla_{\symbfit{r}_{k}} u^{\text{ext}} (\symbfit{r}_{k}) + \sum_{\substack{j=1 \\ j \neq k}}^{n} \left(\nabla_{\symbfit{r}_{k}} u (\symbfit{r}_k, \symbfit{r}_j) \right) \right) \cdot \nabla_{\symbfit{p}_{k}} \right) \phi^{[n]} \\ 
&= \sum_{k=1}^{n} \iint  \left(\nabla_{\symbfit{r}_{k}} u(\symbfit{r}_{k}, \symbfit{r}_{n+1}) \right) \cdot \left(\nabla_{\symbfit{p}_{k}} \phi^{[n+1]}\right) d\symbfit{r}_{n+1} d\symbfit{p}_{n+1}. \label{eq:general_BBGKY}
\end{align}

In this paper, we focus on the equilibrium version of the hierarchy, i.e., 

\begin{equation} 
\begin{aligned}
&\phantom{=}\sum_{k=1}^{n}\left( \frac{\symbfit{p}_k}{m_{k}} \cdot \nabla_{\symbfit{r}_{k}}  - \left(\nabla_{\symbfit{r}_{k}} u^{\text{ext}} (\symbfit{r}_{k}) + \sum_{\substack{j=1 \\ j \neq k}}^{n} \left(\nabla_{\symbfit{r}_{k}} u (\symbfit{r}_k, \symbfit{r}_j) \right) \right) \cdot \nabla_{\symbfit{p}_{k}} \right) \phi^{[n]} \\ 
&= \sum_{k=1}^{n} \iint  \left(\nabla_{\symbfit{r}_{k}} u(\symbfit{r}_{k}, \symbfit{r}_{n+1}) \right) \cdot \left(\nabla_{\symbfit{p}_{k}} \phi^{[n+1]}\right) d\symbfit{r}_{n+1} d\symbfit{p}_{n+1}. 
\end{aligned}\label{eq:equilibrum_BBGKY}
\end{equation} 
One important member of this hierarchy is that of corresponding to the case of $n=1$:
\[\left( \frac{\symbfit{p}_1}{m_{k}} \cdot \nabla_{\symbfit{r}_{1}} - \left(\nabla_{\symbfit{r}_{1}} u^{\text{ext}}(\symbfit{r}_1) \right) \cdot \nabla_{\symbfit{p}_{1}} \right) \phi^{[1]} = \iint \left(\nabla_{\symbfit{r}_{1}} u (\symbfit{r}_1, \symbfit{r}_{2}) \right) \cdot \left(\nabla_{\symbfit{p}_{1}} \phi^{[2]} \right) d\symbfit{r}_{2} d\symbfit{p}_{2}.\]
Here, we assume that $\phi^{[N]}$ is the Boltzmann distribution, i.e.,
\[\phi^{[N]}(\symbfit{r}^N, \symbfit{p}^N) = \frac{1}{h^{3N}N!}\frac{\exp{\left(-\beta H(\symbfit{r}^N, \symbfit{p}^N)\right)}}{Q_N}\]
with the normalization constant (or the canonical partition function) $Q_N$ and $\beta = \frac{1}{k_\mathrm{B}T}$, in which $k_\mathrm{B}$ denotes the Boltzmann constant and $T$ the absolute temperature\footnote{Inclusion of the Planck constant $h$ is a treatment, adopted in \cite{hansen2013theory}, to ensure $Q_{N}$ is dimensionless and consistent with the corresponding quantities of quantum statistical mechanics. We note that discarding the constant will not affect the main results of the paper, while it will alter the interpretation of the distribution.}. Since $\frac{\partial}{\partial p_1} \phi^{[k]} = - \frac{\beta}{m_{k}} \phi^{[k]} \symbfit{p}_1$ for $k=1,2$, the equation is then reduced to a concise form
\begin{align}
\left(k_\mathrm{B}T \nabla_{\symbfit{r}_{1}} + \nabla_{\symbfit{r}_{1}} u^{\text{ext}}(\symbfit{r}_1) \right) \phi^{[1]} = -\iint \left(\nabla_{\symbfit{r}_{1}} u(\symbfit{r}_1, \symbfit{r}_{2}) \right) \phi^{[2]} d\symbfit{r}_2 d\symbfit{p}_2. \label{eq:upper_BBGKY}
\end{align}
Introducing the reduced Boltzmann distribution 
\[\rho^{(n)}(\symbfit{r}^n) = \frac{N!}{(N-n)!} \frac{1}{Z_{N}} \int \exp{\left( -\beta u_{N}(\symbfit{r}^N) \right)} d\symbfit{r}^{N}_{n}\]
with $Z_{N} = \int \exp(-\beta u_{N}(\symbfit{r}^N))d\symbfit{r}^{N}$, the equation \eqref{eq:upper_BBGKY} may be expressed as
\begin{align}
\left( k_\mathrm{B}T \nabla_{\symbfit{r}_{1}} + \left(\nabla_{\symbfit{r}_{1}} u^{\text{ext}}(\symbfit{r}_1) \right)\right) \rho^{(1)}(\symbfit{r}_1) = - \int \left( \nabla_{\symbfit{r}_{1}} u(\symbfit{r}_1, \symbfit{r}_2) \right) \rho^{(2)}(\symbfit{r}_1, \symbfit{r}_2) d\symbfit{r}_2. \label{eq:bbgky_org}
\end{align}

The derivation of the BBGKY hierarchy departs from the assumption that $\phi^{[N]}$ satisfies the Liouville equation. In contrast, the generalized BBGKY hierarchy derived in \cite{168}, called the hyperforce sum rule, is obtained through the first-order functional derivative of the canonical ensemble. The idea does not (explicitly) rely on the Liouville equation and naturally allows the hyperforce sum to include observables, which truly extends the BBGKY hierarchy in the equilibrium setting.

\subsection{Hyperforce sum formula} \label{sec:168}
We consider the same Hamiltonian as in the previous subsection. We set the classical trace operation as $\operatorname{Tr}(\cdot) = \frac{1}{N! h^{dN} }\int d \symbfit{r}^N d \symbfit{p}^N \cdot$. The partition sum $Z$ is represented using the trace as $Z = \operatorname{Tr}e^{-\beta H(\symbfit{r}^{N}, \symbfit{p}^{N})}$. Then, the thermal average $A$ of an observable $\hat{A}(\symbfit{r}^{N}, \symbfit{p}^{N})$ is 
\begin{equation}
A = \langle \hat{A} \rangle = \operatorname{Tr}\left( \hat{A}(\symbfit{r}^{N}, \symbfit{p}^{N}) e^{-\beta H(\symbfit{r}^{N}, \symbfit{p}^{N})}\right)/Z. 
\label{eq:thermal_average_hat{A}}
\end{equation}  

The hyperforce sum rule is obtained through infinitesimal phase space shifts. The full position-resolved phase space shifting operators are defined as  
\[
\sigma (\symbfit{r}) = \sum_{i=1}^{N} \sigma_{i} (\symbfit{r}) = \sum_{i=1}^{N} \left( \delta(\symbfit{r} - \symbfit{r}_{i}) \nabla_{i} + \symbfit{p}_{i} \nabla\delta(\symbfit{r} - \symbfit{r}_{i}) \cdot \nabla_{\symbfit{p}_{i}} \right).
\]
When this operator is applied to the Boltzmann distribution, we get
\begin{align}\sigma(\symbfit{r}) e^{-\beta H} = \beta \hat{\mathbf{F}}(\symbfit{r}) e^{-\beta H}. \label{eq:shift_boltzmann}
\end{align}
Here, 
\begin{align}
    \hat{\mathbf{F}}(\symbfit{r}) = - \nabla \cdot \sum_{i=1}^{N} \frac{\symbfit{p}_{i} \symbfit{p}_{i}^{\operatorname{T}}}{m} \delta(\symbfit{r} - \symbfit{r}_{i}) - \sum_{i=1}^{N} \delta(\symbfit{r} - \symbfit{r}_{i}) \nabla_{i}u_{N}(\symbfit{r}^{N}) - \sum_{i=1}^{N} \delta(\symbfit{r} - \symbfit{r}_{i}) \nabla u^{\text{ext}}(\symbfit{r})
\end{align}
and this operator is obtained through the functional derivative of the canonical transformation induced by a diffeomorphism on the configuration space. 

On the other hand, applying $\sigma (\symbfit{r})$ to an observable $\hat{A}(\symbfit{r}^{N}, \symbfit{p}^{N})$ yields 
\begin{align}\sigma(\symbfit{r})& \hat{A}(\symbfit{r}^{N}, \symbfit{p}^{N}) \\
&= \sum_{i=1}^{N} \delta(\symbfit{r} - \symbfit{r}_{i}) \nabla_{i} \hat{A}(\symbfit{r}^{N}, \symbfit{p}^{N}) + \nabla \cdot \sum_{i=1}^{N} \delta(\symbfit{r} - \symbfit{r}_{i}) \left( \nabla_{\symbfit{p}_{i}} \hat{A}(\symbfit{r}^{N}, \symbfit{p}^{N}) \right)\symbfit{p}_{i}. \label{eq:shift_observable}
\end{align}
Combining \eqref{eq:shift_boltzmann} and \eqref{eq:shift_observable} under the thermal average yields the one-body hyperforce sum rule, expressed as 
\begin{align}
\left\langle \sigma (\symbfit{r}) \hat{A}(\symbfit{r}^{N}, \symbfit{p}^{N}) \right\rangle + \left\langle \hat{A}(\symbfit{r}^{N}, \symbfit{p}^{N}) \beta\, \hat{\mathbf{F}}(\symbfit{r})\right\rangle = 0. \label{eq:original_hyperforce_sum}
\end{align} 
The relation is shown to hold by using the integration by parts. When $\hat{A}$ is a non-vanishing constant function, this relation reduces to the average one-body force density relation in thermal equilibrium setting:
\[\left\langle \beta \hat{\mathbf{F}}(\symbfit{r})\right\rangle = 0,\]
which upon spelling out the three indicidual terms is equivalent to the following form:
\begin{align}
-k_\mathrm{B} T \nabla \rho(\symbfit{r}) - \left\langle \sum_{i=1}^{N} \delta(\symbfit{r} - \symbfit{r}_{i}) \nabla_{i}u_{N}(\symbfit{r}^{N})) \right\rangle - \rho(\symbfit{r}) \nabla u^{\text{ext}}(\symbfit{r}) = 0. \label{eq:bbgky_168}
\end{align}
Here, $\rho(\symbfit{r})$ denotes the one-body density distribution defined as
\[\rho(\symbfit{r}) = \left\langle \sum_{i=1}^{N} \delta(\symbfit{r} - \symbfit{r}_{i}) \right\rangle.\]
We conclude this section by noting that the equation \eqref{eq:bbgky_168} is equivalent to the one derived in \eqref{eq:bbgky_org}.

\section{Case: $\mathbb{R}^{d}$}
\label{sec:Rn} 
In this section, we reformulate the result of \cite{168}, relying on distribution theory on the phase space over $\mathbb{R}^{d}$. All arguments made in this section are presented in the conventional notations adopted in the mathematical literature.

Let $\operatorname{Diff}(\symbb{R}^{d})$ be the set of $C^{\infty}$-diffeomorphisms. 
We introduce three classes of vector fields. 
In this section, we identify with the vector field on $\symbb{R}^{d}$ as a $C^{\infty}$-map on $\symbb{R}^{d}$. 
First, 
let $\symfrak{X}_{\operatorname{Id}}(\symbb{R}^{d})$ be the set of vector fields $\epsilon = (\epsilon^{(1)}, \dots, \epsilon^{(d)})$ on $\symbb{R}^{d}$ such that $\operatorname{Id} + \epsilon \in \operatorname{Diff}(\symbb{R}^{d})$. 
Second, 
let $\symfrak{X}_{0}(\symbb{R}^{d})$ be the set of vector fields with compact support. 
Third, we set 
$\mathfrak{X}_{\operatorname{Id},0}(\mathbb{R}^{d}) 
= \mathfrak{X}_{\operatorname{Id}}(\mathbb{R}^{d}) \cap \mathfrak{X}_{0}(\mathbb{R}^{d})$
and the corresponding class of diffeomorphisms as $\operatorname{Diff}_{\operatorname{Id},0}(\symbb{R}^{d})$. 
For instance, a compactly supported vector field $\epsilon \colon \symbb{R}^{d} \to \symbb{R}^{d}$ such that $\sup_{\symbfit{r} \in \symbb{R}^{d}} \|  \epsilon'(\symbfit{r} ) \| < 1$ is an element in $\mathfrak{X}_{\operatorname{Id},0}(\mathbb{R}^{d})$, where $\|  \epsilon'(\symbfit{r} ) \|$ means the operator norm of the Jacobian matrix of $\epsilon$ at $\symbfit{r}$. 
In the following, we briefly recall $\operatorname{Diff}_{\operatorname{Id},0}(\symbb{R}^{d})$ forms a group by the composition of functions. Although this fact is well-known, we give a detailed proof for the fact. 
\begin{lemma} \label{lem:diffeo}
    $\operatorname{Diff}_{\operatorname{Id}, 0}(\symbb{R}^{d})$ endowed with the composition of functions forms a group.
\end{lemma}
\begin{proof}
    We first show the well-definedness for the composition of diffeomorphisms in $\operatorname{Diff}_{\operatorname{Id}, 0}(\symbb{R}^{d})$. The composition of two functions $\operatorname{Id} + \epsilon_{1}, \operatorname{Id} + \epsilon_{2} \in \operatorname{Diff}_{\operatorname{Id}, 0}(\symbb{R}^{d})$ is 
    \begin{align}
    (\operatorname{Id} + \epsilon_{2}) \circ(\operatorname{Id} + \epsilon_{1})(\symbfit{r}) = \symbfit{r}+ \epsilon_{1}(\symbfit{r}) + \epsilon_{2}(\symbfit{r} + \epsilon_{1}(\symbfit{r})),\ \ \ \symbfit{r} \in \mathbb{R}^{d}. \label{eq:composition}
    \end{align}
    The composition is well-defined for diffeomorphisms in $\operatorname{Diff}_{\operatorname{Id}, 0}(\symbb{R}^{d})$, since this transformation is diffeomorphic by definition and the vector field $\symbfit{r} \mapsto \epsilon_{1}(\symbfit{r}) + \epsilon_{2}(\symbfit{r} + \epsilon_{1}(\symbfit{r}))$ is also smooth with compact support. 

    Next, we will show the inverse operation of diffeomorphisms is well-defined. For any $\operatorname{Id} + \epsilon \in \operatorname{Diff}_{\operatorname{Id}, 0}(\symbb{R}^{d})$, set $\epsilon_{-1}(\symbfit{r}) = -\epsilon \big((\operatorname{Id} + \epsilon)^{-1} (\symbfit{r})\big)$. The vector field has a compact support. Then, the inverse function of $\operatorname{Id} + \epsilon$ is written as $(\operatorname{Id} + \epsilon)^{-1} = \operatorname{Id} + \epsilon_{-1} $.    
    Indeed, 
    \begin{align}
        (\operatorname{Id} + \epsilon_{-1})\big((\operatorname{Id} + \epsilon)(\symbfit{r})\big) 
        &= \symbfit{r} + \epsilon(\symbfit{r}) + \epsilon_{-1}\big(\symbfit{r} + \epsilon(\symbfit{r})\big) \\
        &= \operatorname{Id}(\symbfit{r}) 
    \end{align}
    and we also get
    \begin{align}
        (\operatorname{Id} + \epsilon)\Big( (\operatorname{Id} + \epsilon_{-1})(\symbfit{r}) \Big)
        &= \symbfit{r} + \epsilon_{-1}(\symbfit{r}) + \epsilon \Big(\symbfit{r} + \epsilon_{-1}(\symbfit{r}) \Big) \\
        &= \symbfit{r} - \epsilon\Big((\operatorname{Id} + \epsilon)^{-1} (\symbfit{r})\Big) + \epsilon \Big(\symbfit{r} - \epsilon ((\operatorname{Id} + \epsilon)^{-1}(\symbfit{r}) \Big) \\
        &= \symbfit{r} - \epsilon \Big((\operatorname{Id} + \epsilon)^{-1} (\symbfit{r}) \Big) \\ 
            &\hspace{8mm}+ \epsilon\Big((\operatorname{Id} + \epsilon)\big((\operatorname{Id} + \epsilon)^{-1}(\symbfit{r})\big) - \epsilon \big((\operatorname{Id} + \epsilon)^{-1}(\symbfit{r})\big)\Big) \\ 
        &= \symbfit{r} - \epsilon\Big((\operatorname{Id} + \epsilon)^{-1} (\symbfit{r})\Big)  \\ 
            &\hspace{8mm}+ \epsilon\Big((\operatorname{Id} + \epsilon)^{-1}(\symbfit{r}) + \epsilon\big((\operatorname{Id} + \epsilon)^{-1}(\symbfit{r})\big) - \epsilon \big((\operatorname{Id} + \epsilon)^{-1}(\symbfit{r})\big)\Big) \\ 
        &= \symbfit{r} - \epsilon\Big((\operatorname{Id} + \epsilon)^{-1} (\symbfit{r})\Big) + \epsilon\Big((\operatorname{Id} + \epsilon)^{-1}(\symbfit{r})\Big) \\ 
        &= \operatorname{Id}(\symbfit{r}).
    \end{align}
    The associativity is derived in a straightforward manner, which completes the proof.
\end{proof}

\subsection{Generalized thermal average and hyperforce distribution} 
\label{subsec:thermal_average}
We first consider the phase space $\mathbb R^{dN \times 2}$ on $\symbb R^{dN}$ with linear coordinates $\symbfit{r}^N$ and $\symbfit{p}^{N}$ as introduced in Section \ref{sec:prelim}. 
For each $n \in \{0, 1, \dots, N-1\}$, we also consider its subspace $\symbb{R}^{d(N-n)\times2}$ equipped with the linear coordinates $(\symbfit{r}_{n+1}, \dots, \symbfit{r}_{N})$ and $(\symbfit{p}_{n+1}, \dots, \symbfit{p}_{N})$. 
We will write $\symbfit{r}_{n}^{N}  = (\symbfit{r}_{n+1}, \dots, \symbfit{r}_{N})$ and $\symbfit{p}^{N}_{n} = (\symbfit{p}_{n+1}, \dots, \symbfit{p}_{N})$. 
For $\epsilon \in \mathfrak{X}_{\operatorname{Id},0}(\mathbb{R}^{d})$, we define a diffeomorphism on $\symbb{R}^{dN}$ by
\begin{align}
    \epsilon_{n}(\symbfit{r}^N) = (\symbfit{r}_1, \dots, \symbfit{r}_{n}, (\operatorname{Id} + \epsilon)(\symbfit{r}_{n+1}), \dots, (\operatorname{Id} + \epsilon)(\symbfit{r}_N)). \label{eq:diff_action}
\end{align}
This diffeomorphism naturally gives rise to a diffeomorphism on $\mathbb R^{dN \times 2}$, known as the canonical transformation. The diffeomorphism is expressed as 
\begin{align}
\epsilon_{n, \sharp} : \symbb{R}^{dN\times2} \rightarrow \symbb{R}^{dN\times2}, \ \ \ (\symbfit{r}^{N}, \symbfit{p}^{N}) \mapsto (\hat{\symbfit{r}}^{N}, \hat{\symbfit{p}}^{N}) \label{eq:lift}
\end{align}
in which 
\[
    \hat{\symbfit{r}}_i= 
\begin{cases}
    \symbfit{r}_i \, ,& \text{if } 1 \leq i \leq n\\
    (\operatorname{Id} + \epsilon)(\symbfit{r}_i),              & \text{otherwise}
\end{cases}
\hspace{1mm}, \hspace{5mm}
    \hat{\symbfit{p}}_i= 
\begin{cases}
    \symbfit{p}_i \,,& \text{if } 1 \leq i \leq n\\
    (\symbb{1} + \nabla  \epsilon(\symbfit{r}_i))^{-1}\,\symbfit{p}_i,              & \text{otherwise}
\end{cases} \hspace{1mm}.
\]
Here, $\symbb{1}$ denotes the identity matrix of dimension $d$ and $\nabla \epsilon(\symbfit{r}_{i}) $ is the Jacobian matrix of $\epsilon$ defined as 
\[
\nabla \epsilon(\symbfit{r}_{i}) 
= \begin{pmatrix} \frac{\partial}{\partial r_{1}} \epsilon^{(1)}(\symbfit{r}_{i}) & \cdots & \frac{\partial}{\partial r_{d}} \epsilon^{(1)}(\symbfit{r}_{i}) \\ \vdots  & \ddots & \vdots\\ \frac{\partial}{\partial r_{1}} \epsilon^{(d)}(\symbfit{r}_{i}) & \cdots& \frac{\partial}{\partial r_{d}} \epsilon^{(d)}(\symbfit{r}_{i}) \end{pmatrix}.
\]
We refer to $\epsilon_{n, \sharp}$ as the \textit{lift} of $\epsilon_{n}$ onto $\mathbb R^{dN \times 2}$. We first prove that the pullback by the lift $\epsilon_{n, \sharp}$ induces a well-defined map on $\symscr{S}\left(\symbb{R}^{dN \times 2}\right)$.

\begin{theorem} \label{thm:continous_pullback} 
Let $n \in \{0, 1, \dots, N-1 \}$. 
For any $\epsilon \in \mathfrak{X}_{\operatorname{Id},0}(\mathbb{R}^{d})$, a map $\epsilon_{n, \sharp}^{\ast}$ between $\symscr{S}(\mathbb{R}^{dN \times 2})$ induced by the pullback with the lift $\epsilon_{n, \sharp} \in \operatorname{Diff}_{\symcal{\operatorname{Id}}, 0}(\symbb{R}^{dN \times 2})$
\[\epsilon_{n, \sharp}^{*}: \symscr{S}(\mathbb{R}^{dN \times 2}) \rightarrow \symscr{S}(\mathbb{R}^{dN \times 2})\]
is well-defined. 
\end{theorem}

\begin{proof} We first show the claim for the case of $N=1$ and $n=0$. We prove the argument by employing (a multivariate version of) Faà di Bruno's formula \cite{Brioschi1858,MR1325915}; see Appendix \ref{sec:multi_faddibruno} for the details. The multivariate Faà di Bruno's formula is the explicit expression of the derivative of the composition $f \circ \symbfit{g}$ of two functions $\symbfit{g} = (g^{(1)}, \dots, g^{(m)}): \symbb{R}^{\ell} \rightarrow \symbb{R}^{m}$ and $f: \symbb{R}^{m} \rightarrow \symbb{R}$. This formula, whose complete expression is given in \eqref{eq:multi_faadibruno}, is the finite sum of the following term involving partial derivative of $\symbfit{g}$ and $f$
\begin{align}
    C \cdot f_{\symbfit{\lambda}}  \prod_{j=1}^{s} [\symbfit{g}_{\symbfit{\ell}_{j}}]^{\symbfit{k}_{j}}, \label{eq:faadibruno_maincomp}
\end{align} over multi-indices $\symbfit{\lambda},\,\symbfit{k}_{j} \in \symbb{N}^{m}_{0}$, and $\symbfit{\ell}_{j} \in \symbb{N}^{\ell}_{0}$ with $j$ ranging from $1$ through to a positive integer $s$. Here, $C$ represents the constant involved in each term composing the formula \eqref{eq:multi_faadibruno}. By definition, \eqref{eq:faadibruno_maincomp} is expanded as 
\begin{align}
    \eqref{eq:faadibruno_maincomp}
    &= C \cdot D_{\symbfit{y}}^{\symbfit{\lambda}}f(\symbfit{y}_{0})  \prod_{j=1}^{s} \prod_{i=1}^{\ell} (g_{\symbfit{\ell}_{j}}^{(i)})^{k_{j, i}} \\
    &=  C \cdot D_{\symbfit{y}}^{\symbfit{\lambda}}f(\symbfit{y}_{0})  \prod_{j=1}^{s} \prod_{i=1}^{\ell} (D_{\symbfit{x}}^{\symbfit{\ell}_{j}} g^{(i)}(\symbfit{x}_{0}))^{k_{j, i}}\\
    &= C \cdot \frac{\partial^{|\symbfit{\lambda}|}f}{\partial y_{1}^{\lambda_{1}} \cdots \partial y_{m}^{\lambda_{m}}}(\symbfit{y}_{0})  \prod_{j=1}^{s} \prod_{i=1}^{\ell} \left(\frac{\partial^{|\symbfit{\ell}_{j}|}g^{(i)}}{\partial x_{1}^{\ell_{j, 1}} \cdots \partial x_{\ell}^{\ell_{j, \ell}}} (\symbfit{x}_{0})\right)^{k_{j, i}} \label{eq:expanded_faadibruno}, 
\end{align}
in which $\symbfit{y} = \symbfit{g}(\symbfit{x}) \in \symbb{R}^{m}$ and $\symbfit{y}_{0} = \symbfit{g}(\symbfit{x}_{0}) \in \symbb{R}^{m}$.

We set $\ell = m = d \times 2$ and $\symbfit{x}_{0} = (\symbfit{r}, \symbfit{p}) \in \symbb{R}^{d \times 2}$. We define $\symbfit{g}: \symbb{R}^{d \times 2} \rightarrow \symbb{R}^{d \times 2}$ as 
\[
    \symbfit{g}(\symbfit{x})= \epsilon_{0, \sharp}(\symbfit{x}) = \left(
(\operatorname{Id} + \epsilon)(\symbfit{r}), (\symbb{1} + \nabla \epsilon(\symbfit{r}))^{-1} \symbfit{p} \right).
\]
Set also $f = \phi$ for $\phi \in \mathcal{S}(\mathbb{R}^{d \times 2})$. By plugging $\symbfit{g}$ and $f$ into \eqref{eq:expanded_faadibruno}, we obtain 
\begin{align}
    \eqref{eq:expanded_faadibruno} = C &\cdot \frac{\partial^{|\symbfit{\lambda}|}f}{\partial \hat{r}_{1}^{\lambda_{1}} \cdots \partial \hat{r}_{d}^{\lambda_{d}} \partial \hat{p}_{1}^{\lambda_{d+1}} \cdots \partial \hat{p}_{d}^{\lambda_{2d}}}((\operatorname{Id} + \epsilon)(\symbfit{r}), (\symbb{1} + \nabla \epsilon(\symbfit{r}))^{-1} \symbfit{p}) \\
    &\times \prod_{j=1}^{s} \left\{ \prod_{i=1}^{d} {\underbrace{\left(\frac{\partial^{|\symbfit{\ell}_{j}|}}{\partial r_{1}^{\ell_{j, 1}} \cdots \partial r_{d}^{\ell_{j, d}} \partial p_{1}^{\ell_{j, d+1}} \cdots \partial p_{d}^{\ell_{j, 2d}}} (\operatorname{Id} + \epsilon)(\symbfit{r})_{i}\right)}_{\textrm{(a)}}}^{k_{j, i}} \right. \\ 
    &\times \left. \prod_{k=d+1}^{2d} {\underbrace{\left(\frac{\partial^{|\symbfit{\ell}_{j}|}}{\partial r_{1}^{\ell_{j, 1}} \cdots \partial r_{d}^{\ell_{j, d}} \partial p_{1}^{\ell_{j, d+1}} \cdots \partial p_{d}^{\ell_{j, 2d}}} \left((\symbb{1} + \nabla \epsilon(\symbfit{r}))^{-1} \symbfit{p}\right)_{k-d}\right)}_{\textrm{(b)}}}^{k_{j, k}} \right\}.
\end{align}

Table \ref{tab:derivative_lift} summarizes the asymptotic behaviour of each of the derivatives (a) and (b) with respect to the phase space variables $\symbfit{r}$ and $\symbfit{p}$. From Table \ref{tab:derivative_lift}, we can conclude that \eqref{eq:expanded_faadibruno} is rapidly decreasing because the growth rate of the derivatives (a) and (b) are at most polynomial. In the following, we elaborate the derivation of the behaviour in Table \ref{tab:derivative_lift}; term (a) is bounded because if $\sum_{i=1}^{d}\ell_{j, i} > 0$, (a) becomes the sum of a constant function or the $n$-th order derivative of $\epsilon$ on a compact support,  or $0$ otherwise. 
Similarly, (b) is also bounded over the spatial coordinate $\symbfit{r} \in \symbb{R}^{d}$ when $\sum_{i=1}^{d}\ell_{j, i} \geq 0$. 
Finally, (b) multiplied with a rapidly decreasing function $f_{\symbfit{\lambda}}$ remains rapidly decreasing along $\symbfit{p}$-axis, because the product $\prod_{k=d+1}^{2d}$ of (b) will only involve a polynomial of $\symbfit{p}$ with at most finite order.

\begin{table}[h]
\centering
\caption{Asymptotic behaviour of the (arbitrary-order) derivative of the lift $\epsilon_{n, \sharp}$ with respect to the phase space variable.}
\label{tab:derivative_lift}
\begin{tabular}{c|c|c} 
\toprule
 & $\symbfit{r}$\text{-axis} & $\symbfit{p}$\text{-axis} \\ 
\midrule
\text{(a)} & \text{bounded} & \text{constant} \\  
\text{(b)} & \text{bounded} & \text{polynomial} \\  
\bottomrule
\end{tabular}
\end{table}

Recall that for arbitrary $N$ and $n$, lift $\epsilon_{n,\sharp}$ is defined as the $N$-fold direct product of the identity map, $(\operatorname{Id} + \epsilon)(\symbfit{r})$ and $(\symbb{1} + \nabla \epsilon(\symbfit{r}))^{-1}\,\symbfit{p}$. Therefore, the claim which was proven for the case of $N=1$ and $n=0$ can be naturally extended to the case for arbitrary $N$ and $n$. 
\end{proof}

The following definition gives a thermal average in terms of (tempered) distributions and rapidly decreasing functions. The definition and some basic properties of the Schwartz space and tempered distributions are reviewed in Appendix \ref{app:distribution}. 
The generalized reduced thermal average is an essential observable that describes macroscopic behaviour of microscopic states. 

\begin{definition}[Generalized thermal average] \label{def:thermal_avg_r} 
Let $n \in \{0, 1, \dots, N-1\}$. 
For a tempered distribution $u \in \symscr{S}^{\prime}(\symbb{R}^{d(N-n)\times2})$ on $\symbb{R}^{d(N-n)\times2}$, 
we define a $C(\symbb{R}^{dn \times 2})$-valued tempered distribution 
$K_{n}[u] \in \symscr{S}^{\prime}(\mathbb{R}^{dN \times 2} ; C(\symbb{R}^{dn \times 2}))$ by 
\[
\langle K_{n}[u] , \phi \rangle (\symbfit{r}^{n},\symbfit{p}^{n}) 
= \langle u, \phi(\symbfit{r}^{n} , \cdot , \symbfit{p}^{n} , \cdot) \rangle 
\]
for $\phi \in \symscr{S}(\mathbb{R}^{dN \times 2})$ and $(\symbfit{r}^{n},\symbfit{p}^{n}) \in \symbb{R}^{dn \times 2}$. 
We call the function $\langle K_{n}[u] , \phi \rangle \in C(\symbb{R}^{dn \times 2})$ the {\rm generalized reduced thermal average} of $u$ and $\phi$. 
\end{definition} 

The first example of the generalized thermal average is the case when the tempered distribution is realized through the integration of functions. 

\begin{example}
    One important case of the thermal average is when $u$ is the integration of a tempered continuous function $f$ on $\symbb{R}^{d(N-n) \times 2}$; see Definition \ref{def:tempered_function}. 
    In this case, the generalized thermal average of $u_{f}$ and $\phi \in \symscr{S}(\mathbb{R}^{dN \times 2})$ is given by 
    \[
    \left\langle K_{n}[u_{f}] , \phi \right\rangle (\symbfit{r}^{n}, \symbfit{p}^{n}) = \int_{\symbb{R}^{d(N-n) \times 2}} f(\symbfit{r}^{N}_{n}, \symbfit{p}^{N}_{n}) \phi(\symbfit{r}^{N}, \symbfit{p}^{N}) d\symbfit{r}^{N}_{n} d\symbfit{p}^{N}_{n} 
    \]
    for $(\symbfit{r}^{n}, \symbfit{p}^{n}) \in \symbb{R}^{dn \times 2}$. 
\end{example} 

The other important example is the Dirac delta functional. 

\begin{example} \label{ex:delta_function}   
    For $n \in \{0,1,\dots, N-1\}$ and $i \in \{n+1,\dots , N\}$, 
    the Dirac delta functional $\delta_{(\symbfit{r}_{i}=\symbfit{a})} \in \mathscr{S}'(\symbb{R}^{d(N-n) \times 2})$ on the subspace $\symbfit{r}_{i} = \symbfit{a}$ is defined by 
    \begin{align} 
    &\left\langle \delta_{(\symbfit{r}_{i}=\symbfit{a})} , \phi \right\rangle \\ 
    &\hspace{5mm}= 
    \int_{\symbb{R}^{d(N-n-1)} \times \symbb{R}^{d(N-n)}}\phi(\symbfit{r}_{n+1}, \dots , \symbfit{r}_{i-1} , \overset{i}{\check{\symbfit{a}}}, \symbfit{r}_{i+1} , \dots , \symbfit{r}_{N}, \symbfit{p}_{n}^{N}) d\symbfit{r}_{n+1} \cdots \widehat{d\symbfit{r}_{i}} \cdots d\symbfit{r}_{N} d\symbfit{p}_{n}^{N} \label{eq:delta_def}
    \end{align}
    for $\phi  \in \mathscr{S}(\symbb{R}^{d(N-n) \times 2})$. Here, by $\widehat{d\symbfit{r}_{i}}$ we mean $d\symbfit{r}_{i}$ is excluded. 
    Note that since the delta functional $\delta_{(\symbfit{r}_{i}=\symbfit{a})}$ can be extended to the integrable continuous functions, we often use the same notation $\left\langle \delta_{(\symbfit{r}_{i}=\symbfit{a})} , \phi \right\rangle$ for integrable continuous functions $\phi$ as well.
    The generalized thermal average of $\delta_{(\symbfit{r}_{i}=\symbfit{a})} \in \mathscr{S}'(\symbb{R}^{d(N-n) \times 2})$ and $\phi \in \symscr{S}(\mathbb{R}^{dN \times 2})$ is 
    \begin{align} 
    &\left\langle K_{n}\left[\delta_{(\symbfit{r}_{i}=\symbfit{a})}\right] , \phi \right\rangle (\symbfit{r}^{n}, \symbfit{p}^{n}) \\ 
    &\hspace{5mm}= \int_{\symbb{R}^{d(N-n-1)} \times \symbb{R}^{d(N-n)}} {\phi(\symbfit{r}^{n}, \symbfit{r}_{n+1}, \dots , \overset{i}{\check{\symbfit{a}}} , \dots , \symbfit{r}_{N},\symbfit{p}^{N})} d\symbfit{r}_{n+1} \cdots \widehat{d\symbfit{r}_{i}} \cdots d\symbfit{r}_{N} d\symbfit{p}_{n}^{N} 
    \end{align}  
    for $(\symbfit{r}^{n}, \symbfit{p}^{n}) \in \symbb{R}^{dn \times 2}$. 
    Note that this definition is a mathematical equivalent of that of the Dirac delta function $\delta(\symbfit{r}_{i} - \symbfit{a})$ adopted in physics as
    \[ 
    \left\langle K_{n}\left[\delta_{(\symbfit{r}_{i}=\symbfit{a})}\right] , \phi \right\rangle (\symbfit{r}^{n}, \symbfit{p}^{n}) 
    = 
    \int_{\symbb{R}^{d(N-n)} \times \symbb{R}^{d(N-n)}} \delta( \symbfit{r}_{i} - \symbfit{a} ){\phi(\symbfit{r}^{N},\symbfit{p}^{N})} d\symbfit{r}_{n}^{N} d\symbfit{p}_{n}^{N} 
    \]
    The Dirac delta functional plays a central role to derive useful observables such as the density function such as the one-body density distribution $\rho(\symbfit{r})$ as in Section \ref{sec:168}.  
\end{example}

We now define a functional over $\mathfrak{X}_{\operatorname{Id},0}(\symbb R^{d})$ for each $n \in \{0, 1, \dots, N-1\}$. We first define a functional on $\mathfrak{X}_{\operatorname{Id},0}(\symbb R^{d})$ for given $u \in \symscr{S}^{\prime}(\symbb{R}^{d(N-n)\times2})$ and $\phi \in \symscr{S}(\mathbb{R}^{dN \times 2})$ as follows: 
\[
\mathfrak{X}_{\operatorname{Id},0}(\symbb R^{d}) \rightarrow C\left(\symbb{R}^{dn\times2}\right), \quad 
\epsilon \mapsto \langle \epsilon_{n, \sharp}^{*}K_{n}[u] , \epsilon_{n, \sharp}^{*}\phi \rangle. 
\]
Here, $\epsilon_{n, \sharp}^{*}u$ is the composition of $u$ with $\epsilon_{n, \sharp} \in \operatorname{Diff}_{\symcal{\operatorname{Id}}, 0}(\symbb{R}^{dN \times 2})$, the definition of which is given in Definition \ref{def:composition}. 
Note that the generalized thermal average $\left\langle K_{n}[u], \phi \right\rangle$ is constant regardless of the phase space transformation by the diffeomorphism induced by any $\epsilon \in \mathfrak{X}_{\operatorname{Id},0}(\symbb R^{d})$: 
The definition of $\epsilon_{n, \sharp}^{\ast}$ immediately yields
$\left\langle \epsilon_{n, \sharp}^{*}K_{n}[u] , \epsilon_{n, \sharp}^{*}\phi \right\rangle = \left\langle K_{n}[u] , \phi \right\rangle$ 
for any $\epsilon \in \mathfrak{X}_{\operatorname{Id},0}(\symbb{R}^{d})$.

The property of the functional yields a corresponding property of the derivative:
For any compactly supported vector field $\epsilon \in \symfrak{X}_{0}(\symbb{R}^{d})$ 
and $t \in \symbb{R}$ such that $|t| < \left( \sup_{\symbfit{r} \in \symbb{R}^{d}} \| \epsilon'(\symbfit{r}) \| \right)^{-1}$, 
we have $t\epsilon \in \symfrak{X}_{\operatorname{Id},0}(\symbb{R}^{d})$, and the map 
\begin{equation}
t \mapsto \left\langle (t\epsilon)_{n, \sharp}^{*}K_{n}[u] , (t\epsilon)_{n, \sharp}^{*}\phi \right\rangle
\end{equation}
is well-defined near $t=0$. 
Since the map is constant near $t = 0$ for any $\epsilon \in \mathfrak{X}_{0}(\symbb{R}^{d})$, the derivative of the map at $t = 0$ turns out to also vanish for all $\epsilon \in \mathfrak{X}_{0}(\symbb{R}^{d})$. We formalize the discussion above as a definition and lemma as follows.

\begin{definition} 
\label{def:directional_hyperforce_sum} 
Let $u \in \symscr{S}^{\prime}(\symbb{R}^{d(N-n)\times2})$ and $\phi \in \symscr{S}(\mathbb{R}^{dN \times 2})$. 
We call a functional 
\[
F^{(n)}[u, \phi]: \mathfrak{X}_{0}(\symbb R^{d}) \rightarrow C\left(\symbb{R}^{dn\times2}\right), 
\quad 
\epsilon \mapsto \left. \diff{}{t} \right|_{t=0} \left\langle (t\epsilon)_{n, \sharp}^{*}K_{n}[u] , (t\epsilon)_{n, \sharp}^{*}\phi \right\rangle 
\] 
the \textup{equilibrium distributional hyperforce sum of $u$ and $\phi$} on $\symbb{R}^{d(N-n) \times 2}$.
\end{definition}

\begin{lemma} \label{lem:vanishing_gateaux}
For $u \in \symscr{S}^{\prime}(\symbb{R}^{d(N-n)\times2})$ and $\phi \in \symscr{S}(\mathbb{R}^{dN \times 2})$, 
the equilibrium distributional hyperforce sum $F^{(n)}[u, \phi]$ of $u$ and $\phi$ vanishes on $\mathfrak{X}_{0}(\symbb R^{d})$. We call this vanishing property of $F^{(n)}[u, \phi]$ the \textup{equilibrium distributional hyperforce sum rule}.
\end{lemma}

In the following, we give an alternative expression of $F^{(n)}[u, \phi]$ in Theorem \ref{thm:hyperforce_sum_rewrite}, in a way that it does not involve the derivative with respect to $t$. The derivation makes use of the product rule of the derivative for pairing of tempered distributions. Although the result of Lemma \ref{lem:vanishing_gateaux} is at first glance not indicative from a mathematical point of view and this derivation sounds redundant, it turns out that this seemingly trivial vanishing rule involves a class of physically significant phenomena; Indeed, when $f$ is a tempered smooth function and $\phi$ is the Boltzmann distribution, both the hyperforce sum rule \cite{168} and BBGKY hierarchy \cite{hansen2013theory} can be restored from the equilibrium distributional hyperforce sum rule. Corollary \ref{cor:derivation_hyperforce_sum_rule} and \ref{cor:generalized_168} will give the demonstration and it relies on Theorem \ref{thm:hyperforce_sum_rewrite}. 

We begin recasting the expression of $F^{(n)}[u, \phi]$ with applying the product rule of distributions to the $F^{(n)}[u, \phi]$; Let $u_{t} = (t \epsilon)_{n, \sharp}^{*}K_{n}[u]$ and $\phi_{t} = (t \epsilon)_{n, \sharp}^{*}\phi$ for $t \in \mathbb{R}$. 
By the product rule in Theorem \ref{thm:product_rule}, 
we get the following expression:  
\begin{equation}
\left. \diff{}{t} \right|_{t=0} \left\langle u_{t} , \phi_{t} \right\rangle 
= \left\langle \left. \diff{u_{t}}{t} \right|_{t=0} , \phi \right\rangle + \left\langle u, \left. \diff{\phi_{t}}{t} \right|_{t=0}\right\rangle. 
\label{eq:product_vanish_hyperforce}
\end{equation} 
In order to calculate the right hand side of equation \eqref{eq:product_vanish_hyperforce}, 
we note that 
\begin{align}
\left. \diff{\phi_{t}}{t} \right|_{t=0} \left(\symbfit{r}^{N} , \symbfit{p}^{N} \right) 
&= \left. \diff{}{t} \right|_{t=0} \phi \circ (t\epsilon)_{n,\sharp}\left(\symbfit{r}^{N} , \symbfit{p}^{N} \right)  \\ 
&= \sum_{i=n+1}^{N} 
\left( (\nabla_{\symbfit{r}_{i}}\phi)^{\operatorname{T}}  \epsilon(\symbfit{r}_{i}) 
- (\nabla_{\symbfit{p}_{i}}\phi)^{\operatorname{T}}(\nabla \epsilon)(\symbfit{r}_{i}) \cdot \symbfit{p}_{i} \right). \label{eq:derivative_rapidly_decreasing}
\end{align}
Here, we denote
\begin{equation}
\nabla_{\symbfit{r}_{i}}\phi 
= 
\begin{bmatrix}
\diffp{\phi}{r_{i}^{1}}, \diffp{\phi}{r_{i}^{2}}, \dots , \diffp{\phi}{r_{i}^{d}}  
\end{bmatrix}^{\operatorname{T}} 
\quad \text{and} \quad 
\nabla_{\symbfit{p}_{i}}\phi 
= 
\begin{bmatrix}
\diffp{\phi}{p_{i}^{1}}, \diffp{\phi}{p_{i}^{2}}, \dots , \diffp{\phi}{p_{i}^{d}}  
\end{bmatrix}^{\operatorname{T}}. 
\end{equation}
We also introduce maps $\delta_{i}, \hat{\delta}_{i} \colon \symfrak{X}_{0}(\symbb{R}^{d}) \to C(\symbb{R}^{dN \times 2} , \symbb{R}^{d})$
for $i = n+1, \dots , N$ 
defined by 
\begin{equation}
\delta_{i}(\epsilon)\left(\symbfit{r}^{N} , \symbfit{p}^{N} \right) 
= \epsilon (\symbfit{r}_{i}) 
\quad \text{and} \quad 
\hat{\delta}_{i}(\epsilon)\left(\symbfit{r}^{N} , \symbfit{p}^{N} \right) 
= (\nabla \epsilon)(\symbfit{r}_{i}) \cdot \symbfit{p}_{i}, 
\quad \text{for } \epsilon \in \symfrak{X}_{0}(\symbb{R}^{d}).
\end{equation} 
Then, the equation \eqref{eq:derivative_rapidly_decreasing} may be rewritten as
\begin{equation}
\left. \diff{}{t} \right|_{t=0} \phi \circ (t\epsilon)_{n,\sharp}
= 
\sum_{i=n+1}^{N} 
\left( (\nabla_{\symbfit{r}_{i}}\phi)^{\operatorname{T}} \delta_{i}(\epsilon) - (\nabla_{\symbfit{p}_{i}}\phi)^{\operatorname{T}} \hat{\delta}_{i}(\epsilon) 
\right). 
\end{equation}
Similarly, we have 
\begin{equation}
\left. \diff{}{t} \right|_{t=0} \phi \circ (t\epsilon)_{n,\sharp}^{-1}
= 
\sum_{i=n+1}^{N} 
\left( 
-(\nabla_{\symbfit{r}_{i}}\phi)^{\operatorname{T}}  \delta_{i}(\epsilon) + (\nabla_{\symbfit{p}_{i}}\phi)^{\operatorname{T}}\hat{\delta}_{i}(\epsilon) 
\right). 
\end{equation}
Under the above notation, the first term of the right hand side of equation \eqref{eq:product_vanish_hyperforce} is further expanded as follows:
\begin{align} 
\left\langle \left. \diff{u_{t}}{t} \right|_{t=0} , \phi \right\rangle
&= 
\left\langle \left. \diff{}{t} \right|_{t=0} (t \epsilon)_{n, \sharp}^{*}K_{n}[u], \phi \right\rangle   
= \lim_{h \rightarrow 0} \left\langle \frac{(h\epsilon)_{n, \sharp}^{*}K_{n}[u] -  K_{n}[u]}{h}, \phi \right\rangle  \\
&= \lim_{h \rightarrow 0} \left\langle K_{n}[u],  \frac{ \phi \circ (h \epsilon)_{n, \sharp}^{-1} - \phi}{h} \right\rangle  
= \left\langle K_{n}[u], \left. \diff{}{t} \right|_{t=0} \phi \circ (t\epsilon)_{n, \sharp}^{-1} \right\rangle \\ 
&= \sum_{i=n+1}^{N}\left\langle K_{n}[u],  - (\nabla_{\symbfit{r}_{i}}\phi)^{\operatorname{T}} \delta_{i}(\epsilon) + (\nabla_{\symbfit{p}_{i}}\phi)^{\operatorname{T}} \hat{\delta}_{i}(\epsilon) \right\rangle.  
\end{align}
Similarly, the second term is rewritten as 
\begin{align}
\left\langle u, \left. \frac{d\phi_{t}}{dt} \right|_{t=0}\right\rangle 
&= \left\langle K_{n}[u], \left. \frac{d}{dt} \right|_{t=0} (t \epsilon)_{n, \sharp}^{*}\phi \right\rangle \\ 
&= \sum_{i=n+1}^{N}\left\langle K_{n}[u],  (\nabla_{\symbfit{r}_{i}}\phi)^{\operatorname{T}} \delta_{i}(\epsilon) - (\nabla_{\symbfit{p}_{i}}\phi)^{\operatorname{T}} \hat{\delta}_{i}(\epsilon) 
 \right\rangle. 
\end{align}
Here, we define the \textit{infinitesimal phase shifting operator} 
\begin{equation}
D_{i}(\epsilon) \colon 
\symscr{S}(\symbb{R}^{dN \times 2}) \to \symscr{S}(\symbb{R}^{dN \times 2}) 
\end{equation} 
by 
\begin{equation}
D_{i}(\epsilon)\phi 
= 
\delta_{i}(\epsilon)^{\operatorname{T}} \nabla_{\symbfit{r}_i}\phi 
- \hat{\delta}_{i}(\epsilon)^{\operatorname{T}} \nabla_{\symbfit{p}_{i}}\phi
\end{equation}
for $\phi \in \symscr{S}(\symbb{R}^{dN \times 2})$. 
By using the product rule for derivative, the first term of the shifting operator on the right hand side is rewritten as 
\begin{align}
\delta_{i}(\epsilon)^{\operatorname{T}} \left( \nabla_{\symbfit{r}_{i}}\phi \right)
&= 
\sum_{k=1}^{d} \delta_{i}(\epsilon)^{(k)} \diffp{\phi}{r_{i}^{k}} \\ 
&= 
\sum_{k=1}^{d} \left( \diffp{}{r_{i}^{k}} \left( \delta_{i}(\epsilon)^{(k)}\phi \right) - \diffp{\delta_{i}(\epsilon)^{(k)}}{r_{i}^{k}}\phi  \right). 
\end{align}
Similarly, by the product rule and the fact that $\hat{\delta}_{i}(\epsilon)$ is constant function on $\symbfit{p}_{i}$, the second term is expressed as 
\begin{align}
\hat{\delta}_{i}(\epsilon)^{\operatorname{T}} (\nabla_{\symbfit{p}_{i}}\phi) 
&= 
\sum_{k=1}^{d} \hat{\delta}_{i}(\epsilon)^{(k)} \diffp{\phi}{p_{i}^{k}} \\ 
&= 
\sum_{k=1}^{d} \left( \diffp{}{p_{i}^{k}} \left( \hat{\delta}_{i}(\epsilon)^{(k)}\phi \right) - \diffp{\hat{\delta}_{i}(\epsilon)^{(k)}}{p_{i}^{k}}\phi  \right) \\ 
&= 
\sum_{k=1}^{d} \left( \diffp{}{p_{i}^{k}} \left( \hat{\delta}_{i}(\epsilon)^{(k)}\phi \right) - \diffp{\delta_{i}(\epsilon)^{(k)}}{r_{i}^{k}}\phi  \right). 
\end{align} 
Therefore, for any $u \in \symscr{S}^{\prime}(\symbb{R}^{d(N-n) \times 2})$ and $(\symbfit{r}_{n},\symbfit{p}_{n}) \in \symbb{R}^{dn \times 2}$, we have 
\begin{align}
\langle u , D_{i}(\epsilon)\phi (\symbfit{r}_{n}, \cdot , \symbfit{p}_{n}, \cdot) \rangle 
&= 
\sum_{k=1}^{d} \left\langle u , \left( \diffp{}{r_{i}^{k}} \left( \delta_{i}(\epsilon)^{(k)}\phi \right) 
- \diffp{}{p_{i}^{k}} \left( \hat{\delta}_{i}(\epsilon)^{(k)}\phi \right) \right)(\symbfit{r}_{n}, \cdot , \symbfit{p}_{n}, \cdot) 
\right\rangle \\ 
&= 
-\sum_{k=1}^{d} \left\langle  \delta_{i}(\epsilon)^{(k)}\diffp{u}{r_{i}^{k}} - \hat{\delta}_{i}(\epsilon)^{(k)}\diffp{u}{p_{i}^{k}},  \phi  (\symbfit{r}_{n}, \cdot , \symbfit{p}_{n}, \cdot) 
\right\rangle \\ 
&= 
-\left\langle \left(\delta_{i}(\epsilon)^{\operatorname{T}} \nabla_{\symbfit{r}_i}
- \hat{\delta}_{i}(\epsilon)^{\operatorname{T}} \nabla_{\symbfit{p}_{i}} \right) u , \phi (\symbfit{r}_{n}, \cdot , \symbfit{p}_{n}, \cdot) \right\rangle. 
\end{align} 
Thus the operator $D_{i}(\epsilon)$ can be extended on $\symscr{S}^{\prime}(\symbb{R}^{d(N-n) \times 2})$ as 
\begin{equation} 
\left\langle D_{i}(\epsilon)u , \phi(\symbfit{r}_{n}, \cdot , \symbfit{p}_{n}, \cdot) \right\rangle 
= 
- \left\langle u , D_{i}(\epsilon)\phi (\symbfit{r}_{n}, \cdot , \symbfit{p}_{n}, \cdot) \right\rangle 
\end{equation}
for 
$u \in \symscr{S}^{\prime}(\symbb{R}^{d(N-n) \times 2})$ and $\phi \in \symscr{S}(\symbb{R}^{dN \times 2})$, 
and we obtain
\begin{equation}
\left\langle \left. \diff{u_{t}}{t} \right|_{t=0} , \phi \right\rangle 
= \langle K_{n}[D_{i}(\epsilon)u] , \phi \rangle. 
\end{equation}
From the above discussion, we obtain the following result:

\begin{theorem} \label{thm:hyperforce_sum_rewrite}
For $u \in \symscr{S}^{\prime}(\symbb{R}^{d(N-n)\times2})$ and $\phi \in \symscr{S}(\symbb{R}^{dN\times2})$, we define the \textup{localized hyperforce of $u$ and $\phi$ centered at $i$} to be a continuous-function-valued map 
$F^{(n)}_{i}[u, \phi]: \mathfrak{X}_{0}(\mathbb{R}^{d}) \rightarrow C\left(\symbb{R}^{dn\times2}\right)$ as: 
\begin{align}
F_{i}^{(n)}[u,\phi][\epsilon] = \left\langle K_{n}[D_{i}(\epsilon)u] , \phi \right\rangle + \left\langle K_{n}[u] , D_{i}(\epsilon)\phi \right\rangle. 
\end{align}
Then, the sum of the localized hyperforce coincides with the equilibrium distributional hyperforce sum, i.e., 
\begin{align}
    F^{(n)}[u, \phi] = \sum_{i=n+1}^{N}F_{i}^{(n)}[u, \phi].
\end{align}
\end{theorem} 

By Lemma \ref{lem:vanishing_gateaux}, we have the following. 

\begin{corollary}[Equilibrium distributional hyperforce sum rule] \label{cor:hyperforcesum} The equilibrium distributional hyperforce sum rule is equivalent to 
\[
F_{i}^{(n)}[u, \phi] = 0, \text{\quad for any $i$}.
\]
\end{corollary} 

The result is a generalization of the equilibrium hyperforce sum rule derived in \cite{168}. We will show that in Corollary \ref{cor:derivation_hyperforce_sum_rule} the equilibrium distributional hyperforce sum rule yields the sum rule of \cite{168} and BBGKY hierarchy.

\begin{remark} 
The arguments in this section can be generalized to 
a Banach-space-valued tempered distribution. While such generalization could be useful for applications such as giving another derivation of the Boltzmann equation based on \cite{gerasimenko2013}, we omit the details since it is out of the scope of this paper. 
\end{remark}

\section{Application}
\subsection{Physical interpretation of distributional hyperforce sum formula} \label{sec:physical_hyperforce}
In this subsection, we show that when $u$ in the localized hyperforce $F_{i}^{(n)}[u, \phi]$ is the integration $u_{f}$ of some continuous function $f$, the hyperforce sum is identically zero as a function-valued functional. We also show that the result is consistent with the original result referred as the hyperforce sum rule in \cite{168}. 

For a $\symbb{R} \cup \{\infty\}$-valued function $H$, we define ``singular'' points of $H$ as 
\[
\Delta_{H} 
= \{ (\symbfit{r}^N, \symbfit{p}^N) \in \mathbb{R}^{dN \times 2} \,;\,  H(\symbfit{r}^N, \symbfit{p}^N) = \infty \}.
\]
\begin{definition} \label{def:real_hamiltonian}  
Let $H \in C(\symbb{R}^{dN \times 2}; \symbb{R} \cup \{\infty\})$. 
$H$ is called \textrm{Hamiltonian} when $e^{-H} \in \symscr{S}(\mathbb{R}^{dN \times 2})$. We denote the set of Hamiltonian functions by $\symcal{H}(\symbb{R}^{dN \times 2})$.
\end{definition}
Note that the value of $H \in \mathcal{H}(\symbb{R}^{dN \times 2})$ is the extended real numbers and the continuity at $\symbfit{x}_{0} \in \symbb{R}^{dN \times 2}$ of $H$ is defined as 
\[ 
\lim_{\symbfit{x} \rightarrow \symbfit{x}_{0}} H(\symbfit{x}) = H(\symbfit{x}_{0}). 
\]

The definition indicates that the derivative of $e^{-H}$ is continuously extended to the full domain $\symbb{R}^{dN \times 2}$ regardless that $H$ takes $\infty$ on $\Delta_{H}$. This formulation is crucial to include representative Hamiltonian functions in physical systems, such as repulsive potentials inversely proportional to the spatial distance, as we will detail in the following Example \ref{ex:hamiltonian}. 

\begin{example} \label{ex:hamiltonian}
Our definition of $\symcal{H}(\symbb{R}^{dN \times 2})$ includes Hamiltonian that consists of three contributions: the kinetic energy, an interparticle interaction potential, and an external one-body potential.
\[H(\symbfit{r}^N, \symbfit{p}^N) = \sum_{i=1}^{N} \frac{\symbfit{p}_{i} \symbfit{p}_{i}^{\operatorname{T}}}{2m_{i}}  + u_{N}(\symbfit{r}^N) + \sum_{i=1}^{N}u^\textup{ext}(\symbfit{r}_i).\]

\xhdr{Interparticle potential $u_{N}$}  A typical interparticle potential that makes $e^{-H}$ belong to the Schwartz class is one including the harmonic potential, as we also see, for example, in the classical molecular force fields \cite{charmm, gromacs, amber}. Our definition of Hamiltonian allows the interparticle potential $u_{N}$ in the above expression to include repulsive potentials such Lennard-Jones potential \cite{lennard_jones_second, lennard_jones_first} and the Coulomb potential \cite{halliday2013fundamentals}. 

\xhdr{External potential $u^{\mathup{ext}}$} One example of $u^{\textup{ext}}$ is the field potential in three dimensional polymer system \cite{twistedpolymer}, where a field is exerting a force on each monomer. 
Another example is the total external potential for a thermal system under sedimentation-diffusion equilibrium \cite{whyNoether}, consisting of a gravitational contribution and repulsive contribution associated with a container wall. 
\end{example}

By choosing and fixing a Hamiltonian function $H \in \symcal{H}(\symbb{R}^{dN \times 2})$, we introduce the simple notation for the thermal average for simplicity: For a tempered distribution $u$ on $\symbb{R}^{d(N-n)\times 2}$ and a function $f$ on a subset in $\symbb{R}^{d(N-n)\times 2}$ such that 
the domain of $K_{n}[u]$ can be extended to include the product $fe^{-H}$, we define  
\[ 
\langle uf \rangle_{n} = \left\langle  K_{n}[u] ,  fe^{-H} \right\rangle. 
\] 

Note that for $n=0$, $\langle u \rangle_{0}$ is a distributional generalization of the thermal average in equation \eqref{eq:thermal_average_hat{A}} (up to the temperature and Boltzmann constant).

\begin{example} 
\label{exm:thermal_average_f_{n}}
Suppose that $f$ is a tempered continuous function on $\symbb{R}^{dN\times 2}$.
For a given $(\symbfit{r}^{n} , \symbfit{p}^{n}) \in \symbb{R}^{dn \times 2}$, 
we define $f_{n}(\symbfit{r}_{n}^{N}, \symbfit{p}_{n}^{N}) = f(\symbfit{r}^{n}, \symbfit{r}_{n}^{N}, \symbfit{p}^{n}, \symbfit{p}_{n}^{N})$. 
Then $f_{n}$ is a tempered continuous function on $\symbb{R}^{d(N-n) \times 2}$. 
Then the reduced thermal average 
$u_{f_{n}}$ and $\phi = e^{-H}$ for $H \in \symcal{H}(\symbb{R}^{dN \times 2})$
\[
\left\langle u_{f_{n}} \right\rangle_{n}(\symbfit{r}^{n}, \symbfit{p}^{n}) 
= \int_{\mathbb R^{d(N-n) \times 2}} f_{n}(\symbfit{r}_{n}^{N}, \symbfit{p}_{n}^{N}) e^{-H(\symbfit{r}^N, \symbfit{p}^N)}d\symbfit{r}^{N}_{n} d\symbfit{p}^{N}_{n}
\] 
is finite. 
The finiteness can be assured, for example when $f$ is tempered; see Definition \ref{def:tempered_function}. 
Then, $\left\langle u_{f_{n}} \right\rangle_{0}$ coincides with the standard thermal average (up to the temperature and Boltzmann constant.)
Further, this thermal average also coincides with the canonical partition function (up to the temperature and Boltzmann constant) when $f = 1$. 
\end{example}

We demonstrate that the vanishing phenomenon of the distributional hyperforce shown in Corollary \ref{cor:hyperforcesum} yields the original hyperfoce sum rule \eqref{eq:original_hyperforce_sum} and the equilibrium BBGKY hierarchy as instances; 
let $H$ be a Hamiltonian function in the form of that introduced in Example \ref{ex:hamiltonian} composed by three contributions (over $\symbb{R}^{d}$), the kinetic energy, interparticle interaction potential, and an external one-body potential:
\[
H(\symbfit{r}^N, \symbfit{p}^N) = \sum_{i=1}^{N} \frac{\symbfit{p}_{i}\symbfit{p}_{i}^{\operatorname{T}}}{2m_{i}} + u_{N}(\symbfit{r}^N) + \sum_{i=1}^{N}u^\text{ext}(\symbfit{r}_i).
\]

\begin{lemma} \label{cor:general_derivation_hyperforce_sum_rule} 
Let $\beta = \frac{1}{k_\mathrm{B}T}$, in which $k_\mathrm{B}$ denotes the Boltzmann constant and $T$ the absolute temperature. Assume $u_{N}$ and $u^{\textup{ext}}$ are selected so that $e^{-\beta H(\symbfit{r}^N, \symbfit{p}^N)} \in \symscr{S}(\symbb{R}^{dN \times 2})$ and $f$ is a tempered $C^{1}$ function on $\symbb{R}^{dN \times 2}$. 
Set the following $C(\symbb{R}^{dn \times 2})^{d}$-valued functions $G_{i,k}^{(n)}[f,H]$ on $\symbb{R}^{d}$ for all $i \in \{n+1,\dots, N\}$, $k=1,2,3,4$;
here we denote $\symbb{R}^{d,N}_{n,i}(\symbfit{a}) 
= \symbb{R}^{d(i-n-1)} \times \{ \symbfit{a} \} \times \symbb{R}^{d(N-i)} \times \symbb{R}^{d(N-n)}$ for $\symbfit{a} \in \symbb{R}^{d}$: 
\begin{align}
G_{i,1}^{(n)}[f, H](\symbfit{r}_{i})
&= 
\int_{\symbb{R}^{d,N}_{n,i}(\symbfit{r}_{i})} \left((\nabla_{\symbfit{r}_i}f) \, e^{-\beta H}\right) (\symbfit{r}^{N},\symbfit{p}^{N}) \, 
        d\symbfit{r}_{n+1} \cdots \overset{i}{\widehat{d\symbfit{r}_{i}}} \cdots d\symbfit{r}_{N} d\symbfit{p}_{n}^{N}, \\ 
G_{i,2}^{(n)}[f, H](\symbfit{r}_{i}) 
&= 
\int_{\symbb{R}^{d,N}_{n,i}(\symbfit{r}_{i})} \left(\nabla_{\symbfit{r}_i} \left((\nabla_{\symbfit{p}_i} f)  e^{-\beta H} \right)\right) 
        (\symbfit{r}^{N},\symbfit{p}^{N})  \symbfit{p}_{i} \, 
        d\symbfit{r}_{n+1} \cdots \overset{i}{\widehat{d\symbfit{r}_{i}}} \cdots d\symbfit{r}_{N} d\symbfit{p}_{n}^{N}, \\ 
G_{i,3}^{(n)}[f, H](\symbfit{r}_{i}) 
&= 
\int_{\symbb{R}^{d,N}_{n,i}(\symbfit{r}_{i})} \left(f\,  \nabla_{\symbfit{r}_i} e^{-\beta H} \right) (\symbfit{r}^{N},\symbfit{p}^{N}) \, 
        d\symbfit{r}_{n+1} \cdots \overset{i}{\widehat{d\symbfit{r}_{i}}} \cdots d\symbfit{r}_{N} d\symbfit{p}_{n}^{N}, \\ 
G_{i,4}^{(n)}[f, H](\symbfit{r}_{i}) 
&= - 
\int_{\symbb{R}^{d,N}_{n,i}(\symbfit{r}_{i})} \beta\frac{\symbfit{p}_{i}\symbfit{p}_{i}^{\operatorname{T}}}{m_{i}} \left( \nabla_{\symbfit{r}_i} \left(f e^{-\beta H} \right) \right) (\symbfit{r}^{N},\symbfit{p}^{N}) \, 
        d\symbfit{r}_{n+1} \cdots \overset{i}{\widehat{d\symbfit{r}_{i}}} \cdots d\symbfit{r}_{N} d\symbfit{p}_{n}^{N}.
\end{align}
Then we have 
\begin{equation}
F_{i}^{(n)}[u_{f_{n}}, e^{-\beta H(\symbfit{r}^N, \symbfit{p}^N)}][\epsilon] 
= \sum_{k=1}^{4} \int_{\symbb{R}^{d}}\epsilon(\symbfit{r}_{i})^{\operatorname{T}} G_{i,k}^{(n)}[f, H](\symbfit{r}_{i}) d\symbfit{r}_{i}. 
\end{equation} 
\end{lemma}
\begin{proof} 
Set $\phi(\symbfit{r}^{N}, \symbfit{p}^{N}) = e^{-\beta H(\symbfit{r}^N, \symbfit{p}^N)}$. 
Since $f$ is differentiable, the localized hyperforce $F^{(n)}_{i}[u_{f} , \phi]$ for any $i$ is written as 
\begin{align}
F^{(n)}_{i}[u_{f_{n}} , \phi][\epsilon] 
&= 
\left\langle K_{n}[D_{i}(\epsilon)u_{f_{n}}] , \phi \right\rangle + \left\langle K_{n}[u_{f_{n}}] , D_{i}(\epsilon)\phi \right\rangle \\ 
&= 
\left\langle K_{n}[\delta_{i}(\epsilon)^{\operatorname{T}}\nabla_{\symbfit{r}_{i}}u_{f_{n}}] , \phi \right\rangle 
- \left\langle K_{n}[\hat{\delta}_{i}(\epsilon)^{\operatorname{T}}\nabla_{\symbfit{p}_{i}}u_{f_{n}}] , \phi \right\rangle \\ 
&\hspace{25mm}+ \left\langle K_{n}[u_{f_{n}}] , \delta_{i}(\epsilon)^{\operatorname{T}}\nabla_{\symbfit{r}_{i}}\phi \right\rangle 
- \left\langle K_{n}[u_{f_{n}}] , \hat{\delta}_{i}(\epsilon)^{\operatorname{T}}\nabla_{\symbfit{p}_{i}}\phi \right\rangle. 
\end{align} 
Each term of the right hand side of the above equation can be calculated as follows; here, we apply integration by parts in the second and fourth terms. 
\begin{align} 
&\left\langle K_{n}[\delta_{i}(\epsilon)^{\operatorname{T}} \nabla_{\symbfit{r}_{i}}u_{f_{n}}]  , \phi \right\rangle  (\symbfit{r}^{n} , \symbfit{p}^{n}) \\ 
&\hspace{15mm}= 
\int_{\symbb{R}^{d(N-n) \times 2}} \left( \delta_{i}(\epsilon)^{\operatorname{T}} \nabla_{\symbfit{r}_{i}}f_{n} \right)  (\symbfit{r}^{N}_{n} , \symbfit{p}^{N}_{n}) \phi(\symbfit{r}^{N}, \symbfit{p}^{N}) \, d\symbfit{r}_{n}^{N}d\symbfit{p}_{n}^{N} \\ 
&\hspace{15mm}= 
\int_{\symbb{R}^{d(N-n) \times 2}} \epsilon (\symbfit{r}_{i})^{\operatorname{T}} \left( \left( \nabla_{\symbfit{r}_{i}}f \right)  \phi \right) (\symbfit{r}^{N}, \symbfit{p}^{N}) \, d\symbfit{r}_{n}^{N}d\symbfit{p}_{n}^{N} \\ 
&\hspace{15mm}= 
\int_{\symbb{R}^{d}} \epsilon (\symbfit{r}_{i})^{\operatorname{T}} G_{i,1}^{(n)}[f, H](\symbfit{r}_{i}) d\symbfit{r}_{i}, \\ 
&\left\langle K_{n}[\hat{\delta}_{i}(\epsilon)^{\operatorname{T}}   \nabla_{\symbfit{p}_{i}}u_{f_{n}}]  , \phi \right\rangle (\symbfit{r}^{n} , \symbfit{p}^{n}) \\ 
&\hspace{15mm}= 
\int_{\symbb{R}^{d(N-n) \times 2}} \left( \hat{\delta}_{i}(\epsilon)^{\operatorname{T}} \nabla_{\symbfit{p}_{i}}f_{n} \right)  (\symbfit{r}^{N}_{n} , \symbfit{p}^{N}_{n}) \phi(\symbfit{r}^{N}, \symbfit{p}^{N}) \, d\symbfit{r}_{n}^{N}d\symbfit{p}_{n}^{N} \\ 
&\hspace{15mm}=
\int_{\symbb{R}^{d(N-n) \times 2}} 
\left( \left( \nabla_{\symbfit{r}_{i}}\epsilon \right)(\symbfit{r}_{i}) \cdot \symbfit{p}_{i} \right)^{\operatorname{T}} \left( (\nabla_{\symbfit{p}_{i}}f ) \phi \right) (\symbfit{r}^{N}, \symbfit{p}^{N})  \, d\symbfit{r}_{n}^{N}d\symbfit{p}_{n}^{N} \\ 
&\hspace{15mm}=
\int_{\symbb{R}^{d(N-n) \times 2}} \left( (\nabla_{\symbfit{p}_{i}}f )\phi \right)(\symbfit{r}^{N}, \symbfit{p}^{N})^{\operatorname{T}}(\nabla_{\symbfit{r}_{i}}\epsilon)(\symbfit{r}_{i})  \cdot  \symbfit{p}_{i} \,  d\symbfit{r}_{n}^{N}d\symbfit{p}_{n}^{N} \\ 
&\hspace{15mm}= -
\int_{\symbb{R}^{d(N-n) \times 2}} \epsilon(\symbfit{r}_{i})^{\operatorname{T}}  \left( \nabla_{\symbfit{r}_{i}} \left( (\nabla_{\symbfit{p}_{i}}f )\phi \right) \right)(\symbfit{r}^{N}, \symbfit{p}^{N}) \cdot \symbfit{p}_{i} \, d\symbfit{r}_{n}^{N}d\symbfit{p}_{n}^{N} \\  
&\hspace{15mm}= -
\int_{\symbb{R}^{d}} \epsilon(\symbfit{r}_{i})^{\operatorname{T}} G_{i,2}^{(n)}[f, H](\symbfit{r}_{i}) \,  d\symbfit{r}_{i}, \\ 
&\left\langle K_{n}[u_{f_{n}}] , \delta_{i}(\epsilon)^{\operatorname{T}}\nabla_{\symbfit{r}_{i}}\phi \right\rangle  (\symbfit{r}^{n} , \symbfit{p}^{n}) \\ 
&\hspace{15mm}= 
\int_{\symbb{R}^{d(N-n) \times 2}} f(\symbfit{r}^{N}, \symbfit{p}^{N})  \left( \delta(\epsilon)^{\operatorname{T}} \nabla_{\symbfit{r}_{i}} \phi \right) (\symbfit{r}^{N}, \symbfit{p}^{N}) \, d\symbfit{r}_{n}^{N}d\symbfit{p}_{n}^{N} \\ 
&\hspace{15mm}= 
\int_{\symbb{R}^{d(N-n) \times 2}} f(\symbfit{r}^{N}, \symbfit{p}^{N}) \epsilon(\symbfit{r}_{i})^{\operatorname{T}} (\nabla_{\symbfit{r}_{i}} \phi) (\symbfit{r}^{N}, \symbfit{p}^{N}) \, d\symbfit{r}_{n}^{N}d\symbfit{p}_{n}^{N} \\ 
&\hspace{15mm}= 
\int_{\symbb{R}^{d}} \epsilon(\symbfit{r}_{i})^{\operatorname{T}} G_{i,3}^{(n)}[f, H](\symbfit{r}_{i}) \,  d\symbfit{r}_{i}, \\ 
&\left\langle K_{n}[u_{f_{n}}] , \hat{\delta}_{i}(\epsilon)^{\operatorname{T}}\nabla_{\symbfit{p}_{i}}\phi \right\rangle  (\symbfit{r}^{n} , \symbfit{p}^{n}) \\ 
&\hspace{15mm}= 
\int_{\symbb{R}^{d(N-n) \times 2}} f(\symbfit{r}^{N}, \symbfit{p}^{N})  \left( \hat{\delta}_{i}(\epsilon)^{\operatorname{T}}\nabla_{\symbfit{p}_{i}}\phi \right) (\symbfit{r}^{N}, \symbfit{p}^{N}) \, d\symbfit{r}_{n}^{N}d\symbfit{p}_{n}^{N} \\ 
&\hspace{15mm}= 
\int_{\symbb{R}^{d(N-n) \times 2}} f(\symbfit{r}^{N}, \symbfit{p}^{N})  \left( \left( \nabla_{\symbfit{r}_{i}}\epsilon \right)(\symbfit{r}_{i}) \cdot \symbfit{p}_{i} \right)^{\operatorname{T}} \left( \nabla_{\symbfit{p}_{i}}\phi \right) (\symbfit{r}^{N}, \symbfit{p}^{N}) \, d\symbfit{r}_{n}^{N}d\symbfit{p}_{n}^{N} \\ 
&\hspace{15mm}= 
\int_{\symbb{R}^{d(N-n) \times 2}} \left( f \left(\nabla_{\symbfit{p}_{i}}\phi \right) \right)^{\operatorname{T}} (\symbfit{r}^{N}, \symbfit{p}^{N}) \left( \nabla_{\symbfit{r}_{i}}\epsilon \right)(\symbfit{r}_{i}) \cdot \symbfit{p}_{i}  \, d\symbfit{r}_{n}^{N}d\symbfit{p}_{n}^{N} \\ 
&\hspace{15mm}= -
\int_{\symbb{R}^{d(N-n) \times 2}} \epsilon^{\operatorname{T}} (\symbfit{r}_{i})  \left( \nabla_{\symbfit{r}_{i}} \left(f \left(\nabla_{\symbfit{p}_{i}}\phi \right)\right) \right)(\symbfit{r}^{N}, \symbfit{p}^{N}) \cdot \symbfit{p}_{i}  \, d\symbfit{r}_{n}^{N}d\symbfit{p}_{n}^{N} \\ 
&\hspace{15mm}=  \dfrac{\beta}{m_{i}}
\int_{\symbb{R}^{d(N-n) \times 2}} \epsilon(\symbfit{r}_{i})^{\operatorname{T}} \left( \nabla_{\symbfit{r}_{i}}(f\phi \symbfit{p}_{i})\right) (\symbfit{r}^{N}, \symbfit{p}^{N}) \cdot \symbfit{p}_{i}    \, d\symbfit{r}_{n}^{N}d\symbfit{p}_{n}^{N} \\ 
&\hspace{15mm}= \dfrac{\beta}{m_{i}}
\int_{\symbb{R}^{d(N-n) \times 2}} \epsilon(\symbfit{r}_{i})^{\operatorname{T}}  \symbfit{p}_{i} \symbfit{p}_{i}^{\operatorname{T}} \left( \nabla_{\symbfit{r}_{i}}(f\phi)\right)  (\symbfit{r}^{N}, \symbfit{p}^{N})    \, d\symbfit{r}_{n}^{N}d\symbfit{p}_{n}^{N} \\ 
&\hspace{15mm}= -
\int_{\symbb{R}^{d}} \epsilon(\symbfit{r}_{i})^{\operatorname{T}} G_{i,4}^{(n)}[f, H](\symbfit{r}_{i}) \,  d\symbfit{r}_{i}. 
\end{align}  
Thus, we completed the proof. 
\end{proof} 

By Lemma \ref{cor:general_derivation_hyperforce_sum_rule}, 
we obtain the following reduced hyperforce rule at level $n$: 

\begin{theorem}[Reduced hyperforce rule at level $n$] 
\label{thm:generalized_BBGKY}
    Assume the same condition as in Lemma \ref{cor:general_derivation_hyperforce_sum_rule}. Set
    \begin{equation} 
    G_{i}^{(n)}[f, H](\symbfit{r}_{i}) 
    = 
    \sum_{k=1}^{4} G_{i,k}^{(n)}[f, H](\symbfit{r}_{i}). 
    \end{equation} 
    Then, we have 
    $G_{i}^{(n)}[f, H](\symbfit{r}_{i}) = 0$ 
    for $i=n+1, \cdots , N$.
\end{theorem}
\begin{proof}
    By Corollary \ref{cor:hyperforcesum} and Lemma \ref{cor:general_derivation_hyperforce_sum_rule}, 
    we have 
    \begin{equation} 
    \int_{\symbb{R}^{d}} \epsilon(\symbfit{r}_{i})^{\operatorname{T}} G_{i}^{(n)}[f, H](\symbfit{r}_{i})  d\symbfit{r}_{i} = 0. 
    \end{equation}
    Since the equation holds for any vector fields $\epsilon$ with compact support, 
    we obtain 
    $G_{i}^{(n)}[f, H](\symbfit{r}_{i}) = 0$ 
    for $i=n+1, \cdots , N$. 
\end{proof}

Theorem \ref{thm:generalized_BBGKY} yields the equilibrium BBGKY hierarchy \eqref{eq:equilibrum_BBGKY}. 

\begin{corollary} \label{cor:derivation_hyperforce_sum_rule} 
Assume the same condition as in Lemma \ref{cor:general_derivation_hyperforce_sum_rule}. We also assume $\displaystyle u_{N}(\symbfit{r}^{N}) = \sum_{1 \leq i < j \leq N} u(\symbfit{r}_i, \symbfit{r}_j)$ with $u$ being symmetric. 
Then, for any $n = 1,2,\dots , N-1$ and $k = 1,2,\dots, n$, we have 
\begin{align}
&\left( \frac{\symbfit{p}_{k}}{m_{k}} \cdot \nabla_{\symbfit{r}_{k}}  - \left(\nabla_{\symbfit{r}_{k}} u^{\text{ext}} (\symbfit{r}_{k}) + \sum_{\substack{j=1 \\ j \neq k}}^{n} \left(\nabla_{\symbfit{r}_{k}} u (\symbfit{r}_{k}, \symbfit{r}_{j}) \right) \right) \cdot \nabla_{\symbfit{p}_{k}} \right) \phi^{[n]} \\ 
&= \int_{\symbb{R}^{2d}}  \left(\nabla_{\symbfit{r}_{k}} u(\symbfit{r}_{k}, \symbfit{r}_{n+1}) \right) \cdot \left(\nabla_{\symbfit{p}_{k}} \phi^{[n+1]}\right) d\symbfit{r}_{n+1} d\symbfit{p}_{n+1}. 
\end{align}
Here, $\phi^{[n]}$ is the reduced phase-space distribution function for $\phi = e^{-\beta H}$ defined in Section \ref{sec:bbgky}. Hence, summing up the terms over all $k$ yields the equilibrium BBGKY hierarchy \eqref{eq:equilibrum_BBGKY}. 
\end{corollary} 
\begin{proof}
    Inserting $f = 1$, $G^{(n)}_{i}[f,H](\symbfit{r}_{i})$ is reduced to     
    \begin{equation}
        G^{(n)}_{i}[1, H](\symbfit{r}_{i}) 
        = G^{(n)}_{i,3}[1, H](\symbfit{r}_{i}) + G^{(n)}_{i,4}[1, H](\symbfit{r}_{i}),
    \end{equation}
    in which
    \begin{align}
        G^{(n)}_{i,3}[1, H](\symbfit{r}_{i}) 
        &= \int_{\symbb{R}^{d,N}_{n,i}(\symbfit{r}_{i})} \left(\nabla_{\symbfit{r}_i}\phi \right)(\symbfit{r}^{N},\symbfit{p}^{N}) d\symbfit{r}_{n+1} \cdots \overset{i}{\widehat{d\symbfit{r}_{i}}} \cdots d\symbfit{r}_{N} d\symbfit{p}^{N}_{n}
    \end{align}
    and 
    \begin{align}
        G^{(n)}_{i,4}[1, H](\symbfit{r}_{i}) 
        &= -\int_{\symbb{R}^{d,N}_{n,i}(\symbfit{r}_{i})}  \beta\frac{\symbfit{p}_{i}\symbfit{p}_{i}^{\operatorname{T}}}{m_{i}} \left(\nabla_{\symbfit{r}_i} \phi \right)  (\symbfit{r}^{N},\symbfit{p}^{N}) d\symbfit{r}_{n+1} \cdots \overset{i}{\widehat{d\symbfit{r}_{i}}} \cdots d\symbfit{r}_{N} d\symbfit{p}^{N}_{n}.
    \end{align}

    We assume $i = n + 1$.  
    Then, 
    \begin{align}
        G^{(n)}_{n+1,3}[1, H](\symbfit{r}_{n+1}) 
        &= \int_{\symbb{R}^{d,N}_{n,n+1}(\symbfit{r}_{n+1})} \left(\nabla_{\symbfit{r}_{n+1}}\phi \right)(\symbfit{r}^{N},\symbfit{p}^{N}) \overset{n+1}{\widehat{d\symbfit{r}_{n+1}}} d\symbfit{r}_{n+2} \cdots d\symbfit{r}_{N} d\symbfit{p}^{N}_{n}\\
        &= - \beta \underbrace{\int_{\symbb{R}^{d,N}_{n,n+1}(\symbfit{r}_{n+1})}  \nabla_{\symbfit{r}_{n+1}} u^{\text{ext}}(\symbfit{r}_{n+1}) \phi(\symbfit{r}^{N}, \symbfit{p}^{N})\ d\symbfit{r}_{n+1}^{N} d\symbfit{p}^{N}_{n}}_{I} \\ 
        &\hspace{10mm}- \beta \sum_{\substack{j = 1 \\ j \neq n+1}}^{N} \underbrace{\int_{\symbb{R}^{d,N}_{n,n+1}(\symbfit{r}_{n+1})} \nabla_{\symbfit{r}_{n+1}} u(\symbfit{r}_{n+1}, \symbfit{r}_{j}) \phi(\symbfit{r}^{N}, \symbfit{p}^{N})\ d\symbfit{r}_{n+1}^{N} d\symbfit{p}^{N}_{n}}_{I_{j}}. \\ 
    \end{align}
    Here, we introduce the following notations for $1 \leq k \leq n+1$:
    \begin{align}
        Z_{n+1} &= \displaystyle \int_{\symbb{R}^{d}} e^{-\beta \symbfit{p}_{n+1}^{\operatorname{T}}\symbfit{p}_{n+1}} d\symbfit{p}_{n+1}, \\
        \displaystyle \phi_{\widehat{\symbfit{p}_{k}}}(\symbfit{r}^{N}, \symbfit{p}^{N}) &= e^{-\beta \left( H(\symbfit{r}^{N}, \symbfit{p}^{N}) - \symbfit{p}_{k}^{\operatorname{T}} \symbfit{p}_{k} \right)}, \\
        \phi^{[n+1]}_{\widehat{\symbfit{p}_{k}}}(\symbfit{r}^{n+1}, \symbfit{p}^{n+1}) 
        &=  \frac{N!}{(N-n-1)!}\int_{\symbb{R}^{d(N-n-1) \times 2}} \phi_{\widehat{\symbfit{p}_{k}}}(\symbfit{r}^{N}, \symbfit{p}^{N}) d\symbfit{r}^{N}_{n+1} d\symbfit{p}_{n+1}^{N}, \\  
        \phi^{[n+1]}  (\symbfit{r}^{n+1}, \symbfit{p}^{n+1}) &= \frac{N!}{(N-n-1)!}\int_{\symbb{R}^{d(N-n-1)\times 2}} \phi(\symbfit{r}^{N}, \symbfit{p}^{N}) d\symbfit{r}^{N}_{n+1} d\symbfit{p}^{N}_{n+1}.
    \end{align} 
    The integral $I$ in the first term is further expanded as 
    \begin{align}
        I &= \int_{\symbb{R}^{d,N}_{n,n+1}(\symbfit{r}_{n+1})}  \nabla_{\symbfit{r}_{n+1}} u^{\text{ext}}(\symbfit{r}_{n+1}) \phi(\symbfit{r}^{N}, \symbfit{p}^{N})\ d\symbfit{r}_{n+1}^{N} d\symbfit{p}^{N}_{n} \\
        &=   Z_{n+1} \int_{\symbb{R}^{d(N-n-1) \times 2}}\nabla_{\symbfit{r}_{n+1}} u^{\text{ext}}(\symbfit{r}_{n+1}) \phi_{\widehat{\symbfit{p}_{n+1}}}(\symbfit{r}^{N}, \symbfit{p}^{N})\ d\symbfit{r}_{n+1}^{N} d\symbfit{p}^{N}_{n+1} \\
        &=  Z_{n+1} \frac{(N-n-1)!}{N!} \nabla_{\symbfit{r}_{n+1}} u^{\text{ext}}(\symbfit{r}_{n+1})  \phi^{[n+1]}_{\widehat{\symbfit{p}_{n+1}}}(\symbfit{r}^{n+1}, \symbfit{p}^{n+1}).        
    \end{align}
    Similarly, $I_{j}$ is rewritten as 
    \begin{align}
        I_{j} &=  \int_{\symbb{R}^{d,N}_{n,n+1}(\symbfit{r}_{n+1})}   \nabla_{\symbfit{r}_{n+1}} u(\symbfit{r}_{n+1}, \symbfit{r}_{j}) \phi(\symbfit{r}^{N}, \symbfit{p}^{N})\ d\symbfit{r}_{n+1}^{N} d\symbfit{p}^{N}_{n} \qquad (\text{for \, $j=1,2,\dots,n$}) \\
        &=  Z_{n+1}\int_{\symbb{R}^{d(N-n-1) \times 2}}  \nabla_{\symbfit{r}_{n+1}} u(\symbfit{r}_{n+1}, \symbfit{r}_{j})  \phi_{\widehat{\symbfit{p}_{n+1}}}(\symbfit{r}^{N}, \symbfit{p}^{N})\ d\symbfit{r}_{n+1}^{N} d\symbfit{p}^{N}_{n+1} \\
        &=  Z_{n+1} \frac{(N-n-1)!}{N!}  \nabla_{\symbfit{r}_{n+1}} u(\symbfit{r}_{n+1}, \symbfit{r}_{j}) \phi^{[n+1]}_{\widehat{\symbfit{p}_{n+1}}}(\symbfit{r}^{n+1}, \symbfit{p}^{n+1})   
    \end{align} 
    and 
    \begin{align}
        I_{j} &= 
        \int_{\symbb{R}^{d,N}_{n,n+1}(\symbfit{r}_{n+1})} \nabla_{\symbfit{r}_{n+1}} u(\symbfit{r}_{n+1}, \symbfit{r}_{j}) \phi(\symbfit{r}^{N}, \symbfit{p}^{N})\ d\symbfit{r}_{n+1}^{N} d\symbfit{p}^{N}_{n} \qquad (\text{for \, $j = n+2, \dots, N$}) \\ 
        &= Z_{n+1}\int_{\symbb{R}^{d(N-n-1) \times 2}} \nabla_{\symbfit{r}_{n+1}} u(\symbfit{r}_{n+1}, \symbfit{r}_{j}) \phi_{\widehat{\symbfit{p}_{n+1}}}(\symbfit{r}^{N}, \symbfit{p}^{N})\ d\symbfit{r}_{n+1}^{N} d\symbfit{p}^{N}_{n+1} \\ 
        &= Z_{n+1}\int_{\symbb{R}^{d(N-n-1) \times 2}} \nabla_{\symbfit{r}_{n+1}} u(\symbfit{r}_{n+1}, \symbfit{r}_{n+2}) \phi_{\widehat{\symbfit{p}_{n+1}}}(\symbfit{r}^{N}, \symbfit{p}^{N}) d\symbfit{r}^{N}_{n+1} d\symbfit{p}^{N}_{n+1} \\ 
        &= Z_{n+1}\dfrac{(N-n-2)!}{N!} \int_{\symbb{R}^{2d}} \nabla_{\symbfit{r}_{n+1}} u(\symbfit{r}_{n+1}, \symbfit{r}_{n+2}) \phi^{[n+2]}_{\widehat{\symbfit{p}_{n+1}}}(\symbfit{r}^{n+2}, \symbfit{p}^{n+2}) d\symbfit{r}_{n+2} d\symbfit{p}_{n+2}. 
    \end{align}
    On the other hand, using the integration by parts, $G^{(n)}_{i,4}[1, H](\symbfit{r}_{n+1})$ is written as
    \begin{align}
        G^{(n)}_{i,4}[1, H](\symbfit{r}_{n+1}) 
        &= \int_{\symbb{R}^{d,N}_{n,n+1}(\symbfit{r}_{n+1})} \left(\left(-\beta\frac{\symbfit{p}_{n+1}\symbfit{p}_{n+1}^{\operatorname{T}}}{m_{n+1}} \right)\left(\nabla_{\symbfit{r}_{n+1}} \phi \right) \right)\left(\symbfit{r}^{N}, \symbfit{p}^{N}\right) d\symbfit{r}_{n+1}^{N} d\symbfit{p}^{N}_{n}\\
        &= \int_{\symbb{R}^{d,N}_{n,n+1}(\symbfit{r}_{n+1})} \left(\symbfit{p}_{n+1} \left(\nabla_{\symbfit{p}_{n+1}}^{\operatorname{T}} \nabla_{\symbfit{r}_{n+1}} \phi \right) \right)\left(\symbfit{r}^{N}, \symbfit{p}^{N}\right) d\symbfit{r}_{n+1}^{N} d\symbfit{p}^{N}_{n}\\
        &= - \int_{\symbb{R}^{d,N}_{n,n+1}(\symbfit{r}_{n+1})} \left(\nabla_{\symbfit{r}_{n+1}} \phi \right)(\symbfit{r}^{N}, \symbfit{p}^{N}) d\symbfit{r}_{n+1}^{N} d\symbfit{p}^{N}_{n} \\
        &= - Z_{n+1} \int_{\symbb{R}^{d(N-n-1) \times 2}} \left(\nabla_{\symbfit{r}_{n+1}} \phi_{\widehat{\symbfit{p}_{n+1}}} \right)(\symbfit{r}^{N}, \symbfit{p}^{N}) d\symbfit{r}_{n+1}^{N} d\symbfit{p}^{N}_{n+1} \\
        &= - Z_{n+1} \frac{(N-n-1)!}{N!} \left( \nabla_{\symbfit{r}_{n+1}} \phi^{[n+1]}_{\widehat{\symbfit{p}_{n+1}}} \right) (\symbfit{r}^{n+1} , \symbfit{r}^{n+1}), \label{eq:bbgky_second_final}
    \end{align}    
    Here, $\nabla_{\symbfit{p}_{n+1}}^{\operatorname{T}}$ is a ``divergence'' operator along $\symbfit{p}_{n+1}$. 
    Therefore, noting that $G^{(n)}[1, H] = 0$ rearranging 
    \begin{equation}
    -\dfrac{1}{\beta}G_{n+1,4}^{(n)}[1,H](\symbfit{r}_{n+1}) + I + \sum_{j=1}^{n}I_{j} = -\sum_{j=n+2}^{N} I_{j},
    \end{equation} 
    we have
    \begin{align}
    &\left\{\frac{1}{\beta}\nabla_{\symbfit{r}_{n+1}} +  \left(\nabla_{\symbfit{r}_{n+1}} u^{\text{ext}}(\symbfit{r}_{n+1}) + \sum_{j=1}^{n} \nabla_{\symbfit{r}_{n+1}} u(\symbfit{r}_{n+1}, \symbfit{r}_{j})  \right) \right\}\phi^{[n+1]}_{\widehat{\symbfit{p}_{n+1}}}(\symbfit{r}^{n+1}, \symbfit{p}^{n+1}) \\
    &\hspace{20mm}= - \int_{\symbb{R}^{2d}} \nabla_{\symbfit{r}_{n+1}} u(\symbfit{r}_{n+1}, \symbfit{r}_{n+2}) \phi^{[n+2]}_{\widehat{\symbfit{p}_{n+1}}}(\symbfit{r}^{n+2}, \symbfit{p}^{n+2}) d\symbfit{r}_{n+2} d\symbfit{p}_{n+2}.
    \end{align}
    Notice that the equation does not involve $\symbfit{p}_{n+1}$. Therefore, this equation is equivalent to 
    \begin{align}
    &\left\{\frac{\symbfit{p}_{n+1}}{m_{n+1}} \cdot \nabla_{\symbfit{r}_{n+1}} -  \left(\nabla_{\symbfit{r}_{n+1}} u^{\text{ext}}(\symbfit{r}_{n+1}) + \sum_{j=1}^{n} \nabla_{\symbfit{r}_{n+1}} u(\symbfit{r}_{n+1}, \symbfit{r}_{j})  \right) \cdot \frac{\beta\,\symbfit{p}_{n+1}}{m_{n+1}} \right\}\phi^{[n+1]}_{\widehat{\symbfit{p}_{n+1}}}(\symbfit{r}^{n+1}, \symbfit{p}^{n+1}) \\
    &\hspace{20mm}= \frac{\beta\,\symbfit{p}_{n+1}}{m_{n+1}} \cdot \int_{\symbb{R}^{2d}} \nabla_{\symbfit{r}_{n+1}} u(\symbfit{r}_{n+1}, \symbfit{r}_{n+2}) \phi^{[n+2]}_{\widehat{\symbfit{p}_{n+1}}}(\symbfit{r}^{n+2}, \symbfit{p}^{n+2})   d\symbfit{r}_{n+2} d\symbfit{p}_{n+2}.
    \end{align}
    Since
    \[ 
    \phi^{[n+1]}(\symbfit{r}^{n+1}, \symbfit{p}^{n+1}) = \left( e^{-\beta\frac{\symbfit{p}_{k}^{\operatorname{T}}\symbfit{p}_{k}}{2m_{k}}}\right)\phi^{[n+1]}_{\widehat{\symbfit{p}_{k}}}(\symbfit{r}^{n+1}, \symbfit{p}^{n+1})
    \]
    for any $k = 1,2,\dots , n+1$, 
    the above equation is equivalent to 
    \begin{align}
    &\left\{\frac{\symbfit{p}_{n+1}}{m_{n+1}} \cdot \nabla_{\symbfit{r}_{n+1}} -  \left(\nabla_{\symbfit{r}_{n+1}} u^{\text{ext}}(\symbfit{r}_{n+1}) + \sum_{j=1}^{n} \nabla_{\symbfit{r}_{n+1}} u(\symbfit{r}_{n+1}, \symbfit{r}_{j})  \right) \cdot \frac{\beta\,\symbfit{p}_{n+1}}{m_{n+1}} \right\}\phi^{[n+1]}(\symbfit{r}^{n+1}, \symbfit{p}^{n+1}) \\
    &\hspace{20mm}= \frac{\beta\,\symbfit{p}_{n+1}}{m_{n+1}} \cdot \int_{\symbb{R}^{2d}} \nabla_{\symbfit{r}_{n+1}} u(\symbfit{r}_{n+1}, \symbfit{r}_{n+2}) \phi^{[n+2]}(\symbfit{r}^{n+2}, \symbfit{p}^{n+2})   d\symbfit{r}_{n+2} d\symbfit{p}_{n+2}.
    \end{align}
    Furthermore, noting that $\displaystyle \nabla_{\symbfit{p}_{k}} \phi^{[n+1]}(\symbfit{r}^{n+1}, \symbfit{p}^{n+1}) = - \frac{\beta}{m_{k}}\symbfit{p}_{k} \phi^{[n+1]}(\symbfit{r}^{n+1}, \symbfit{p}^{n+1})$ 
    for any $k = 1,2,\dots , n+1$, we have 
    \begin{align}
    &\left\{\frac{\symbfit{p}_{n+1}}{m_{n+1}} \cdot \nabla_{\symbfit{r}_{n+1}} -  \left(\nabla_{\symbfit{r}_{n+1}} u^{\text{ext}}(\symbfit{r}_{n+1}) + \sum_{j=1}^{n} \nabla_{\symbfit{r}_{n+1}} u(\symbfit{r}_{n+1}, \symbfit{r}_{j})  \right) \cdot \nabla_{\symbfit{p}_{n+1}}\right\}\phi^{[n+1]}(\symbfit{r}^{n+1}, \symbfit{p}^{n+1}) \\
    &\hspace{20mm}= \int_{\symbb{R}^{2d}} \nabla_{\symbfit{r}_{n+1}} u(\symbfit{r}_{n+1}, \symbfit{r}_{n+2}) \cdot \nabla_{\symbfit{p}_{n+1}} \phi^{[n+2]}(\symbfit{r}^{n+2}, \symbfit{p}^{n+2})   d\symbfit{r}_{n+2} d\symbfit{p}_{n+2}.
    \end{align} 
    Finally, the derivation up until the above equation also holds for the case of linear coordinates $(\symbfit{r}_{k}, \symbfit{r}_{n+2}, \dots, \symbfit{r}_{N})$ and $(\symbfit{p}_{k}, \symbfit{p}_{n+2}, \dots, \symbfit{p}_{N})$ for any $k = 1, 2, \dots, n+1$. Therefore, for any $k = 1, 2, \dots, n+1$, 
    the following formula also holds: 
    \begin{align}
    &\left\{\frac{\symbfit{p}_{k}}{m_{k}} \cdot \nabla_{\symbfit{r}_{k}} -  \left(\nabla_{\symbfit{r}_{k}} u^{\text{ext}}(\symbfit{r}_{k}) + \sum_{\substack{j=1 \\  j \neq k}}^{n+1} \nabla_{\symbfit{r}_{k}} u(\symbfit{r}_{k}, \symbfit{r}_{j})  \right) \cdot \nabla_{\symbfit{p}_{k}}\right\}\phi^{[n+1]}(\symbfit{r}^{n+1}, \symbfit{p}^{n+1}) \\
    &\hspace{20mm}= \int_{\symbb{R}^{2d}} \nabla_{\symbfit{r}_{k}} u(\symbfit{r}_{k}, \symbfit{r}_{n+2}) \cdot \nabla_{\symbfit{p}_{k}} \phi^{[n+2]}(\symbfit{r}^{n+2}, \symbfit{p}^{n+2})  d\symbfit{r}_{n+2} d\symbfit{p}_{n+2}.
    \end{align}
    Summing the equation over all $k$ (and replacing $n+1$ with $n$) yields the equilibrium BBGKY formula \eqref{eq:equilibrum_BBGKY}. 
\end{proof}

\begin{corollary} 
\label{cor:generalized_168}
    Define a map $G^{(n)}[f, H]: \symbb{R}^{d} \rightarrow C(\symbb{R}^{dn \times 2})^{d}$ by 
    \[
    G^{(n)}[f, H](\symbfit{r}) = \sum_{i=1}^{N} G^{(n)}_{i}[f, H](\symbfit{r}).
    \]
    Then, we have $G^{(n)}[f, H](\symbfit{r}) = 0$, which is a generalization of the (equilibrium) hyperforce sum rule \eqref{eq:original_hyperforce_sum}; 
    Namely, the original hyperforce sum rule \eqref{eq:original_hyperforce_sum}
    is restored by setting $n=0$. 
\end{corollary}
\begin{proof}
    The formula $G^{(n)}[f, H](\symbfit{r}) = 0$ is directly obtained by Theorem \ref{thm:generalized_BBGKY}. 
    
    Assume $n=0$ and $f = \hat{A}$. 
    We may rewrite $G^{(0)}[\hat{A}, H](\symbfit{r})$ as
    \begin{align}
        G^{(0)}[\hat{A}, H](\symbfit{r}) 
        &=  \sum_{i=1}^{N} \left(G^{(0)}_{i, 1}[\hat{A}, H](\symbfit{r}) + G^{(0)}_{i, 2}[\hat{A}, H](\symbfit{r}) \right) \\ 
        &\hspace{25mm}+ \sum_{i=1}^{N} \left(G^{(0)}_{i, 3}[\hat{A}, H](\symbfit{r}) + G^{(0)}_{i, 4}[\hat{A}, H](\symbfit{r}) \right).
    \end{align}
    Recalling two symbols $\sigma(\symbfit{r})$ and $\hat{\symbf{F}}(\symbfit{r})$ from Section \ref{sec:168} and Example \ref{ex:delta_function}: 
    \begin{align}
    \sigma (\symbfit{r}) 
    &= \sum_{i=1}^{N} \sigma_{i} (\symbfit{r}) = \sum_{i=1}^{N} \left( \delta(\symbfit{r} - \symbfit{r}_{i}) \nabla_{\symbfit{r}_{i}} + \symbfit{p}_{i} \nabla\delta(\symbfit{r} - \symbfit{r}_{i}) \cdot \nabla_{\symbfit{p}_{i}} \right)  \\ 
    &= \sum_{i=1}^{N} \left( \delta_{(\symbfit{r}_{i} = \symbfit{r})} \nabla_{\symbfit{r}_{i}} - \left( \symbfit{p}_{i}^{\operatorname{T}} \nabla_{\symbfit{r}_{i}} \delta_{(\symbfit{r}_{i} = \symbfit{r})} \right) \nabla_{\symbfit{p}_{i}} \right) 
    \shortintertext{and} \\ 
    \hat{\mathbf{F}}(\symbfit{r}) 
    &= \sum_{i=1}^{N} \hat{\mathbf{F}}_{i}(\symbfit{r})\\
    &= -\sum_{i=1}^{N} \left( \nabla \cdot \frac{\symbfit{p}_{i} \symbfit{p}_{i}^{\operatorname{T}}}{m} \delta(\symbfit{r} - \symbfit{r}_{i}) + \delta(\symbfit{r} - \symbfit{r}_{i}) \nabla_{\symbfit{r}_{i}}u_{N}(\symbfit{r}^{N}) + \delta(\symbfit{r} - \symbfit{r}_{i}) \nabla u^{\text{ext}}(\symbfit{r})\right) \\ 
    &= \sum_{i=1}^{N} \left(  \underbrace{ \frac{\symbfit{p}_{i} \symbfit{p}_{i}^{\operatorname{T}}}{m} \nabla_{\symbfit{r}_{i}}  \delta_{(\symbfit{r}_{i} = \symbfit{r})} }_{\hat{\symbf{F}}_{i,1}(\symbfit{r})} - \underbrace{\delta_{(\symbfit{r}_{i} = \symbfit{r})} \left( \nabla_{\symbfit{r}_{i}}u_{N}(\symbfit{r}^{N}) + \nabla u^{\text{ext}}(\symbfit{r}) \right) }_{\hat{\symbf{F}}_{i,2}(\symbfit{r})}  \right).  
    \end{align}  
    For each $i$, the left and right hand sides are represented as follows:
    \begin{align}
    &\phantom{=*}G^{(0)}_{i, 1}[\hat{A}, H](\symbfit{r}) + G^{(0)}_{i, 2}[\hat{A}, H](\symbfit{r})  \\
    &= \int_{\symbb{R}^{d,N}_{0,i}(\symbfit{r})} \left( \left(\nabla_{\symbfit{r}_{i}}\hat{A} \right) \, e^{-\beta H}\right) (\symbfit{r}_{1}, \dots , \overset{i}{\check{\symbfit{r}}}, \dots , \symbfit{r}_{N}, \symbfit{p}^{N}) \, 
    d\symbfit{r}_{1} \cdots \overset{i}{\widehat{d\symbfit{r}}} \cdots d\symbfit{r}_{N} d\symbfit{p}^{N} \\ 
    &\hspace{10mm}+ \int_{\symbb{R}^{d,N}_{0,i}(\symbfit{r})} \left(\nabla_{\symbfit{r}_{i}} \left( \left( \nabla_{\symbfit{p}_{i}} \hat{A} \right)  e^{-\beta H} \right)\right) 
    (\symbfit{r}_{1}, \dots , \overset{i}{\check{\symbfit{r}}}, \dots , \symbfit{r}_{N}, \symbfit{p}^{N}) \cdot \symbfit{p}_{i} \, 
    d\symbfit{r}_{1} \cdots \overset{i}{\widehat{d\symbfit{r}}} \cdots d\symbfit{r}_{N} d\symbfit{p}^{N} \\
    &= 
    \left\langle \delta_{(\symbfit{r}_{i} = \symbfit{r})} , 
    \left( \nabla_{\symbfit{r_{i}}} \hat{A} \right)  e^{-\beta H} 
    \right\rangle 
    + \left\langle \delta_{(\symbfit{r}_{i} = \symbfit{r})} , \nabla_{\symbfit{r}_{i}} \left( \left(  \nabla_{\symbfit{p}_{i}} \hat{A} \right)  e^{-\beta H} \right) \symbfit{p}_{i} \right\rangle \\ 
    &= 
    \left\langle \delta_{(\symbfit{r}_{i} = \symbfit{r})}  \left( \nabla_{\symbfit{r}_{i}} \hat{A} \right)  ,  e^{-\beta H} \right\rangle 
    - \left\langle \left( \symbfit{p}_{i}^{\operatorname{T}}  \nabla_{\symbfit{r}_{i}} \delta_{(\symbfit{r}_{i} = \symbfit{r})} \right)  \nabla_{\symbfit{p}_{i}} \hat{A} ,  e^{-\beta H} \right\rangle \\ 
    &= \left\langle \sigma_{i} (\symbfit{r}) \hat{A} \right\rangle_{0}. 
    \end{align}  
    Similarly, 
    \begin{align}
    &\phantom{=}G^{(0)}_{i, 3}[\hat{A}, H](\symbfit{r}) + G^{(0)}_{i, 4}[\hat{A}, H](\symbfit{r})  \\
    &= \int_{\symbb{R}^{d,N}_{0,i}(\symbfit{r})} \left( \hat{A}  \nabla_{\symbfit{r}_{i}} e^{-\beta H} \right) (\symbfit{r}_{1}, \dots , \overset{i}{\check{\symbfit{r}}}, \dots , \symbfit{r}_{N}, \symbfit{p}^{N}) \, 
    d\symbfit{r}_{1} \cdots \overset{i}{\widehat{d\symbfit{r}}} \cdots d\symbfit{r}_{N} d\symbfit{p}^{N} \\ 
    &\hspace{10mm}- \int_{\symbb{R}^{d,N}_{0,i}(\symbfit{r})} \beta\frac{\symbfit{p}_{i}\symbfit{p}_{i}^{\operatorname{T}}}{m} \left( \nabla_{\symbfit{r}_{i}} \left( \hat{A} e^{-\beta H} \right) \right) (\symbfit{r}_{1}, \dots , \overset{i}{\check{\symbfit{r}}}, \dots , \symbfit{r}_{N}, \symbfit{p}^{N}) \, 
    d\symbfit{r}_{1} \cdots \overset{i}{\widehat{d\symbfit{r}}} \cdots d\symbfit{r}_{N} d\symbfit{p}^{N}\\
    &= - \beta \int_{\symbb{R}^{d,N}_{0,i}(\symbfit{r})} \left( \hat{A}\left(\nabla_{\symbfit{r}_{i}} u_{N} + \nabla u^{\text{ext}} (\symbfit{r}) \right) e^{-\beta H} \right) (\symbfit{r}_{1}, \dots , \overset{i}{\check{\symbfit{r}}}, \dots , \symbfit{r}_{N}, \symbfit{p}^{N}) \, 
    d\symbfit{r}_{1} \cdots \overset{i}{\widehat{d\symbfit{r}}} \cdots d\symbfit{r}_{N} d\symbfit{p}^{N} \\ 
    &\hspace{10mm}- \beta\int_{\symbb{R}^{d,N}_{0,i}(\symbfit{r})} \frac{\symbfit{p}_{i}\symbfit{p}_{i}^{\operatorname{T}}}{m} \left( \nabla_{\symbfit{r}_{i}} \left( \hat{A} e^{-\beta H} \right) \right) (\symbfit{r}_{1}, \dots , \overset{i}{\check{\symbfit{r}}}, \dots , \symbfit{r}_{N}, \symbfit{p}^{N}) \, 
    d\symbfit{r}_{1} \cdots \overset{i}{\widehat{d\symbfit{r}}} \cdots d\symbfit{r}_{N} d\symbfit{p}^{N} \\
    &=- \beta \left\langle  \delta_{(\symbfit{r}_{i} = \symbfit{r})} ,    \hat{A} \left( \nabla_{\symbfit{r}_{i}}u_{N} + \left( \nabla u^{\text{ext}}  \right)(\symbfit{r}_{i}) \right) e^{-\beta H}    \right\rangle     
    - \beta \left\langle \delta_{(\symbfit{r}_{i} = \symbfit{r})} , \frac{\symbfit{p}_{i}\symbfit{p}_{i}^{\operatorname{T}}}{m}    \left( \nabla_{\symbfit{r}_{i}} \left( \hat{A} e^{-\beta H} \right) \right)  \right\rangle \\
    &=- \beta \left\langle   \delta_{(\symbfit{r}_{i} = \symbfit{r})} ,  \hat{A}  \left( \nabla_{\symbfit{r}_{i}}u_{N} + \left( \nabla u^{\text{ext}} \right)(\symbfit{r}_{i}) \right)  e^{-\beta H}   \right\rangle   
    + \beta \left\langle     \frac{\symbfit{p}_{i}\symbfit{p}_{i}^{\operatorname{T}}}{m} \left( \nabla_{\symbfit{r}_{i}}  \delta_{(\symbfit{r}_{i} = \symbfit{r})} \right)  \hat{A} , e^{-\beta H}     \right\rangle \\ 
    &= -\left\langle \beta \hat{\symbf{F}}_{i,2}(\symbfit{r}) \hat{A} \right\rangle_{0} + \left\langle \beta \hat{\symbf{F}}_{i,1}(\symbfit{r})\hat{A} \right\rangle_{0}.
    \end{align}
    Therefore, we obtain
    \begin{align}
        G^{(0)}[\hat{A}, H](\symbfit{r}) 
        &= \sum_{i=1}^{N} \left(G^{(0)}_{i, 1}[\hat{A}, H](\symbfit{r}) + G^{(0)}_{i, 2}[\hat{A}, H](\symbfit{r}) \right) \\ 
        &\hspace{20mm}+ \sum_{i=1}^{N} \left(G^{(0)}_{i, 3}[\hat{A}, H](\symbfit{r}) + G^{(0)}_{i, 4}[\hat{A}, H](\symbfit{r}) \right)\\
        &= \sum_{i=1}^{N} \left( \left\langle \sigma_{i} (\symbfit{r}) \hat{A} \right\rangle_{0} + \left\langle \beta \hat{\symbf{F}}_{i,1}(\symbfit{r})\hat{A} \right\rangle_{0} -\left\langle \beta \hat{\symbf{F}}_{i,2}(\symbfit{r})\hat{A} \right\rangle_{0}  \right).
    \end{align} 
    Thus, $G^{(0)}[\hat{A}, H](\symbfit{r}) = 0$ implies equation \eqref{eq:original_hyperforce_sum}. 
\end{proof}

\subsection{Distributional hyperforce sum rules for periodic boundary systems}
\label{sec:compact}

As another application of the theory developed in Section \ref{sec:Rn}, we give a formulation of hyperforce sum rules for systems with the periodic boundary condition. 
The theory with periodic boundary conditions will yield in parallel to, except some minor differences, the derivation of that of Euclidean space. The results can translate to a torus-equivalent of the hyperforce sum rules.

Let $\{ \symbfit{v}_{1}, \dots , \symbfit{v}_{d} \}$ be the basis for $\symbb{R}^{d}$ and we set 
\begin{equation}
\Gamma = \left\{ \sum_{k=1}^{d}a_{k}\symbfit{v}_{k} \,;\, a_{k} \in \symbb{Z} \right\} \quad \text{and} \quad 
\Lambda = \left\{ \sum_{k=1}^{d}a_{k}\symbfit{v}_{k} \,;\, 0 \leq a_{k} \leq 1 \right\}. 
\end{equation} 
We call that a function $f$ on $\symbb{R}^{d}$ is \textit{$\Gamma$-periodic} 
if $f(\symbfit{x} + \gamma) = f(\symbfit{x})$ holds for any $\symbfit{x} \in \symbb{R}^{d}$ and $\gamma \in \Gamma$. 

Let $\symfrak{X}^{\Gamma}(\symbb{R}^{d})$ be the set of $\Gamma$-periodic vector fields on $\symbb{R}^{d}$. From now on, a vector field $\epsilon = (\epsilon^{(1)}, \dots , \epsilon^{(d)}) \colon \symbb{R}^{d} \to \symbb{R}^{d}$ 
is called $\Gamma$-periodic if each component $\epsilon^{(i)} \colon \symbb{R}^{d} \to \symbb{R}$ is  $\Gamma$-periodic. 
Any $\epsilon \in \mathfrak{X}_{\operatorname{Id},0}(\mathbb{R}^{d}) \cap \symfrak{X}^{\Gamma}(\symbb{R}^{d})$ induces the diffeomorphism $\epsilon_{n,\sharp} \colon \symbb{R}^{dN \times 2} \to \symbb{R}^{dN \times 2}$ as introduced in Section \ref{sec:Rn}. 
Then, $\epsilon_{n,\sharp}$ is $\Gamma$-equivariant in configuration space, that is,  
\begin{equation}
\epsilon_{n,\sharp}(\symbfit{r}^{N} + \gamma^{N} , \symbfit{p}^{N}) 
= 
\epsilon_{n,\sharp}(\symbfit{r}^{N} , \symbfit{p}^{N}) + (\gamma^{N},0)
\end{equation}
for any 
$(\symbfit{r}^{N} , \symbfit{p}^{N}) \in \symbb{R}^{dN \times 2}$ 
and $\gamma^{N} = (\gamma_{1}, \dots , \gamma_{N}) \in \Gamma^{N}$. 

In Theorem \ref{thm:continous_pullback_compact} below, we show a $\Gamma$-periodic equivalent of Theorem \ref{thm:continous_pullback}. 
Let $n \in \{ 0,1,\dots , N-1 \}$. 
We call that a function $f$ on $\symbb{R}^{d(N-n) \times 2}$ 
is $\Gamma$-periodic in configuration space if 
the equation $f \left(\symbfit{r}_{n}^{N} + \gamma^{N-n}, \symbfit{p}_{n}^{N}\right) = f \left(\symbfit{r}_{n}^{N}, \symbfit{p}_{n}^{N}\right)$ holds for any 
$(\symbfit{r}_{n}^{N} , \symbfit{p}_{n}^{N}) \in \symbb{R}^{d(N-n) \times 2}$ 
and $\gamma^{N-n} = (\gamma_{1}, \dots , \gamma_{N-n}) \in \Gamma^{N-n}$. 
We also denote by 
$\symscr{S}_{\textup{per}}(\symbb{R}^{dN \times 2})$ 
the set of rapidly decreasing functions on $\symbb{R}^{dN \times 2}$ 
satisfying $\Gamma$-periodic in configuration space.

\begin{theorem} \label{thm:continous_pullback_compact} 
Let $n \in \{0, 1, \dots, N-1 \}$. 
For any $\epsilon \in \mathfrak{X}_{\operatorname{Id},0}(\mathbb{R}^{d}) \cap \symfrak{X}^{\Gamma}(\symbb{R}^{d})$, 
a map $\epsilon_{n, \sharp}^{\ast}$ between $\symscr{S}_{\textup{per}}(\symbb{R}^{dN \times 2})$ induced by the pullback 
\[
\epsilon_{n, \sharp}^{*}: \symscr{S}_{\textup{per}}(\symbb{R}^{dN \times 2}) \rightarrow \symscr{S}_{\textup{per}}(\symbb{R}^{dN \times 2})
\]
is well-defined. 
\end{theorem}

\begin{proof}  
We only have to prove that a function $\epsilon_{n, \sharp}^{*}\phi$ is $\Gamma$-periodic in configuration space. 
In fact, for any $(\symbfit{r}^{N}, \symbfit{p}^{N}) \in \symbb{R}^{dN \times 2}$ and $\gamma^{N} \in \Gamma^{N}$, we have 
\begin{align}
(\epsilon_{n, \sharp}^{*}\phi) \left( \symbfit{r}^{N} + \gamma^{N}, \symbfit{p}^{N} \right)  
&= \phi \left( \epsilon_{n,\sharp}(\symbfit{r}^{N} + \gamma^{N}, \symbfit{p}^{N}) \right) 
= \phi \left( \epsilon_{n,\sharp}(\symbfit{r}^{N}, \symbfit{p}^{N}) + (\gamma^{N},0) \right) \\ 
&= \phi \left( \epsilon_{n,\sharp}(\symbfit{r}^{N}, \symbfit{p}^{N}) \right) 
= (\epsilon_{n, \sharp}^{*}\phi)  \left( \symbfit{r}^{N}  , \symbfit{p}^{N} \right). 
\end{align}
\end{proof}

\begin{definition}[$\Gamma$-periodic reduced thermal average] 
\label{def:thermal_avg_cotangent} 
Let $n \in \{0, 1, \dots, N-1\}$. 
For a tempered continuous function $f \in C(\symbb{R}^{dN \times 2})$ 
satisfying $\Gamma$-periodic in configuration space 
and $\phi \in \symscr{S}_{\textup{per}}(\symbb{R}^{dN \times 2})$, 
we define a continuous function on $\symbb{R}^{dn \times 2}$  
by 
\[
K_{n}[f,\phi] (\symbfit{r}^{n},\symbfit{p}^{n}) 
= \int_{\Lambda^{N-n} \times \symbb{R}^{d(N-n)}} f(\symbfit{r}^{N}  , \symbfit{p}^{N}) \phi(\symbfit{r}^{N} ,  \symbfit{p}^{N} ) \, d\symbfit{r}_{n}^{N}d\symbfit{p}_{n}^{N}.  
\]
We call $K_{n}[f,\phi]$ the \textrm{$\Gamma$-periodic generalized reduced thermal average} of $f$ and $\phi$. 
\end{definition} 

Similar to the case of the Euclidean space, a functional over $\mathfrak{X}^{\Gamma}(\symbb{R}^{d})$ defined with any pair of 
a $\Gamma$-periodic tempered continuous function $f \in C(\symbb{R}^{dN \times 2})$ in configuration space 
and $\phi \in \symscr{S}_{\textup{per}}(\symbb{R}^{dN \times 2})$ vanishes for any $n$;   
Define the pullback 
\begin{equation}
\tilde{\epsilon}_{n,\sharp}(f) 
= f \circ \epsilon_{n,\sharp}  \cdot \left| \det \epsilon_{n,\sharp}' \right|
\end{equation} 
of $f$; see Example \ref{exm:pullback_function} for the background of the definition. 
Then, a map   
\[
\mathfrak{X}_{\operatorname{Id},0}(\mathbb{R}^{d}) \cap \symfrak{X}^{\Gamma}(\symbb{R}^{d}) \rightarrow C\left( \symbb{R}^{dn \times 2} \right), \quad 
\epsilon \mapsto K_{n}[\tilde{\epsilon}_{n,\sharp}(f),  \epsilon_{n,\sharp}^{\ast}\phi]
\] 
is constant for each $n \in \{0, 1, \dots, N-1\}$ because of the change of variable formula. Hence, we obtain the following map constant near $t=0$ for any $\epsilon \in \mathfrak{X}^{\Gamma}(\symbb{R}^{d})$:
\begin{equation}
t \mapsto K_{n} \left[ \tilde{(t\epsilon)}_{n, \sharp}(f) , (t\epsilon)_{n, \sharp}^{*}\phi \right].
\end{equation}

\begin{definition} 
\label{def:directional_hyperforce_sum_compact} 
Let $f$ be a $\Gamma$-periodic tempered continuous function on $\symbb{R}^{dN \times 2}$ in configuration space 
and $\phi \in \symscr{S}_{\textup{per}}(\symbb{R}^{dN \times 2})$.  
We call a map 
\[
F^{(n)}[f, \phi]: \mathfrak{X}^{\Gamma}(\symbb{R}^{d}) \rightarrow C\left( \symbb{R}^{dn \times 2} \right), 
\quad 
\epsilon \mapsto \left. \diff{}{t} \right|_{t=0} K_{n} \left[ \tilde{(t\epsilon)}_{n, \sharp}(f) , (t\epsilon)_{n, \sharp}^{*}\phi \right] 
\] 
the \textup{$\Gamma$-periodic equilibrium distributional hyperforce sum of $f$ and $\phi$}.
\end{definition}

\begin{lemma} \label{lem:vanishing_gateaux_compact}
Let $f$ be a $\Gamma$-periodic tempered continuous function on $\symbb{R}^{dN \times 2}$ in configuration space 
and $\phi \in \symscr{S}_{\textup{per}}(\symbb{R}^{dN \times 2})$.  
The $\Gamma$-periodic equilibrium distributional hyperforce sum $F^{(n)}[f, \phi]$ of $f$ and $\phi$ vanishes on $\mathfrak{X}^{\Gamma}(\symbb{R}^{d})$. 
We call this vanishing property of $F^{(n)}[f, \phi]$ the \textup{$\Gamma$-periodic equilibrium distributional hyperforce sum rule}. 
\end{lemma}

Finally, we will show that the results in Section \ref{sec:physical_hyperforce} can be extended to the $\Gamma$-periodic setting. 
We introduce corresponding Hamiltonian and their instances.

\begin{definition} \label{def: compact_hamiltonian}
Denote by $\symcal{H}_{\textup{per}}(\symbb{R}^{dN \times 2})$ 
the set of Hamiltonians satisfying $\Gamma$-periodicity in configuration space.  
\end{definition}

In the sequel, we assume that $H \in \symcal{H}_{\textup{per}}(\symbb{R}^{dN \times 2})$ 
satisfies in the form 
\[
H(\symbfit{r}^N, \symbfit{p}^N) 
= \sum_{i=1}^{N} \frac{\symbfit{p}_i^{\operatorname{T}} \symbfit{p}_i}{2m_{i}}  + u_{N}(\symbfit{r}^N) + \sum_{i=1}^{N}u^\text{ext}(\symbfit{r}_i).
\]

\begin{example} 
One representative physical system with the periodic condition is the (monatomic) ideal gas system that consists of $N$ particles on tori, where particles do not interact with each other and they have no internal structures. A typical Hamiltonian in this case is  
\[
H(\symbfit{r}^{N}, \symbfit{p}^{N}) = \sum_{i=1}^{N} \frac{\symbfit{p}_i^{\operatorname{T}} \symbfit{p}_i}{2m_{i}}  . 
\]
The Hamiltonian yields a finite normalization constant, also known as the partition function, of the Boltzmann factor $e^{-H(\symbfit{r}^N, \symbfit{p}^N)}$.

One class of $u_{N}$ is found in the Weeks-Chandler-Andersen (WCA) system \cite{weeks_chandler_andersen_1971}, which cuts the standard Lennard-Jones pair potential at its minimum and subsequently shifts the potential by adding a constant such that the minimum is lifted to zero: 
\[
u_{N}(\symbfit{r}^{N}) = \sum_{1 \leq i < j \leq N}u(r_{i,j}), \quad
u(r) 
= 
\begin{cases}
4 \epsilon \left[ (r/\sigma)^{-12} - (r/\sigma)^{-6}\right] + \epsilon, & (r < 2^{1/6} \sigma)  \\ 
0, & (r \geq 2^{1/6} \sigma).
\end{cases}
\]
Here, $r_{i, j} = || \symbfit{r}_{i} - \symbfit{r}_{j} ||$, $\sigma (> 0)$ reflects the particle radius and $\epsilon$ is the energy depth of the LJ potential well.
The corresponding Hamiltonian (with the kinetic term) gives a finite partition function $Z$ in periodic boundary systems, while it diverges in the Euclidean space. Similarly, the following potential used to study systems with finite range, repulsive potentials jam \cite{hern_silbert_liu_nagel_2003} is another such instance:
\[
u(r_{i, j}) 
= 
\begin{cases}
\epsilon \left( 1 - r_{i, j}/\sigma_{i,j}\right)^{\alpha}/\alpha, & (r_{i, j} <  \sigma_{i, j})  \\ 
0, & (r_{i, j} \geq \sigma_{i, j}).
\end{cases}
\]
Here, $\sigma_{ij}$ is the sum of the radii of particles $i$ and $j$ and $\epsilon$ is the characteristic energy scale of the interaction.

All the systems introduced in Example \ref{ex:hamiltonian} are also legitimate when the $\Gamma$-periodicity is imposed on the Hamiltonian functions. 
\end{example}

Let $\beta = \frac{1}{k_\mathrm{B}T}$, in which $k_\mathrm{B}$ denotes the Boltzmann constant and $T$ the absolute temperature. 
Similar to Lemma \ref{cor:general_derivation_hyperforce_sum_rule}, 
we assume that 
$u_{N}$ and $u^{\textup{ext}}$ are selected so that $H \in \symcal{H}_{\textup{per}}(\symbb{R}^{dN \times 2})$, and $f$ is a $\Gamma$-periodic tempered $C^{1}$ function on $\symbb{R}^{dN \times 2}$ in configuration space. 
Set the following $C(\symbb{R}^{dn \times 2})^{d}$-valued functions 
$G_{i,k}^{(n)}[f,H]$ on $\symbb{R}^{d}$ for all $i \in \{n+1,\dots, N\}$, $k=1,2,3,4$; 
here we denote $\symbb{R}^{d,N}_{n,i,\textup{per}}(\symbfit{a}) 
= \Lambda^{d(i-n-1)} \times \{ \symbfit{a} \} \times \Lambda^{d(N-i)} \times \symbb{R}^{d(N-n)}$ for $\symbfit{a} \in \symbb{R}^{d}$: 
\begin{align}
G_{i,1}^{(n)}[f, H](\symbfit{r}_{i})
&= 
\int_{\symbb{R}^{d,N}_{n,i,\textup{per}}(\symbfit{r}_{i})} \left((\nabla_{\symbfit{r}_i}f) \, e^{-\beta H}\right) (\symbfit{r}^{N},\symbfit{p}^{N}) \, 
        d\symbfit{r}_{n+1} \cdots \overset{i}{\widehat{d\symbfit{r}_{i}}} \cdots d\symbfit{r}_{N} d\symbfit{p}_{n}^{N}, \\ 
G_{i,2}^{(n)}[f, H](\symbfit{r}_{i}) 
&= 
\int_{\symbb{R}^{d,N}_{n,i,\textup{per}}(\symbfit{r}_{i})} \left(\nabla_{\symbfit{r}_i} \left((\nabla_{\symbfit{p}_i} f)  e^{-\beta H} \right)\right) 
        (\symbfit{r}^{N},\symbfit{p}^{N})  \symbfit{p}_{i} \, 
        d\symbfit{r}_{n+1} \cdots \overset{i}{\widehat{d\symbfit{r}_{i}}} \cdots d\symbfit{r}_{N} d\symbfit{p}_{n}^{N}, \\ 
G_{i,3}^{(n)}[f, H](\symbfit{r}_{i}) 
&= 
\int_{\symbb{R}^{d,N}_{n,i,\textup{per}}(\symbfit{r}_{i})} \left(f\,  \nabla_{\symbfit{r}_i} e^{-\beta H} \right) (\symbfit{r}^{N},\symbfit{p}^{N}) \, 
        d\symbfit{r}_{n+1} \cdots \overset{i}{\widehat{d\symbfit{r}_{i}}} \cdots d\symbfit{r}_{N} d\symbfit{p}_{n}^{N}, \\ 
G_{i,4}^{(n)}[f, H](\symbfit{r}_{i}) 
&= - 
\int_{\symbb{R}^{d,N}_{n,i,\textup{per}}(\symbfit{r}_{i})} \beta\frac{\symbfit{p}_{i}\symbfit{p}_{i}^{\operatorname{T}}}{m_{i}} \left( \nabla_{\symbfit{r}_i} \left(f e^{-\beta H} \right) \right) (\symbfit{r}^{N},\symbfit{p}^{N}) \, 
        d\symbfit{r}_{n+1} \cdots \overset{i}{\widehat{d\symbfit{r}_{i}}} \cdots d\symbfit{r}_{N} d\symbfit{p}_{n}^{N}.
\end{align}

Then, we obtain the following $\Gamma$-periodic version of the hyper force sum rule and their corollaries. Their proofs are identical to those of Theorem \ref{thm:generalized_BBGKY}, Corollary \ref{cor:derivation_hyperforce_sum_rule}, and Corollary \ref{cor:generalized_168}.

\begin{theorem}[$\Gamma$-periodic reduced hyperforce rule at level $n$] 
\label{thm:generalized_BBGKY_compact}
    Set
    \begin{equation} 
    G_{i}^{(n)}[f, H](\symbfit{r}_{i}) 
    = 
    \sum_{k=1}^{4} G_{i,k}^{(n)}[f, H](\symbfit{r}_{i}). 
    \end{equation} 
    Then, we have 
    $G_{i}^{(n)}[f, H](\symbfit{r}_{i}) = 0$ 
    for $i=n+1, \cdots , N$.
\end{theorem}

\begin{corollary}[$\Gamma$-periodic BBGKY hierarchy at level $n$] 
\label{cor:derivation_hyperforce_sum_rule_compact} 
Assume $\displaystyle u_{N}(\symbfit{r}^{N}) = \sum_{1 \leq i < j \leq N} u(\symbfit{r}_i, \symbfit{r}_j)$ with $u$ being symmetric. 
Then, for any $n = 1,2,\dots , N-1$ and $k = 1,2,\dots, n$, we have 
\begin{align}
&\left( \frac{\symbfit{p}_{k}}{m_{k}} \cdot \nabla_{\symbfit{r}_{k}}  - \left(\nabla_{\symbfit{r}_{k}} u^{\text{ext}} (\symbfit{r}_{k}) + \sum_{\substack{j=1 \\ j \neq k}}^{n} \left(\nabla_{\symbfit{r}_{k}} u (\symbfit{r}_{k}, \symbfit{r}_{j}) \right) \right) \cdot \nabla_{\symbfit{p}_{k}} \right) \phi^{[n]} \\ 
&= \int_{\Lambda \times \symbb{R}^{d}}  \left(\nabla_{\symbfit{r}_{k}} u(\symbfit{r}_{k}, \symbfit{r}_{n+1}) \right) \cdot \left(\nabla_{\symbfit{p}_{k}} \phi^{[n+1]}\right) d\symbfit{r}_{n+1} d\symbfit{p}_{n+1}. 
\end{align}
Here, we denote 
\begin{equation}
\phi^{[n]}  (\symbfit{r}^{n}, \symbfit{p}^{n}) = \frac{N!}{(N-n)!}\int_{\Lambda^{N-n} \times \symbb{R}^{d(N-n)}} \phi(\symbfit{r}^{N}, \symbfit{p}^{N}) d\symbfit{r}^{N}_{n} d\symbfit{p}^{N}_{n}.
\end{equation}
\end{corollary} 

\begin{corollary}[$\Gamma$-periodic hyperforce sum rule at level $n$] 
\label{cor:generalized_168_compact}
    Define a map $G^{(n)}[f, H]: X \rightarrow C(\symbb{R}^{dn \times 2})^{d}$ by 
    \[
    G^{(n)}[f, H](\symbfit{r}) = \sum_{i=1}^{N} G^{(n)}_{i}[f, H](\symbfit{r}).
    \]
    Then, we have $G^{(n)}[f, H](\symbfit{r}) = 0$.
\end{corollary}

\section*{Summary and Outlook}
The present paper gives a distributional formulation of the equilibrium hyperforce sum rule \cite{PhysRevLett.133.217101, 168}. We introduced the generalized thermal average as the pairing of tempered distributions and rapidly decreasing functions. We defined a functional that is constant on a subspace of the vector fields, which leads to the vanishing property of the derivative of the functional along the vector fields. Furthermore, we showed that the Leibniz rule of the derivative is the distributional equivalent of the original hyperforce sum formula. Indeed, we demonstrated that the Leibniz rule gives rise to the hyperforce sum rules and the equilibrium BBGKY hierarchy. As for another application, we also applied our argument to the class of periodic tempered functions and rapidly decreasing functions, and derived the hyperforce sum rules on tori as an immediate corollary.

In future work, it is natural to consider the extension of our results to general manifolds in terms of the canonical cotangent bundles over manifolds. It is also interesting to investigate categorical interpretation of our result through the lens of functorial Poisson structure \cite{MarsdenMorrisonWeinstein1984, weinstein1996lagrangian}. Other potential generalization is to give a distributional interpretation to the case of equilibrium fluid mixtures \cite{gauge_invariance_equilibrium_mixture}, non-equilibrium BBGKY hierarchies \cite{171}, and generalized non-equilibrium BBGKY hierarchies \cite{biagetti_generalized_bbgky}. Turning the theoretical result into real applications is also intriguing. One such direction is to investigate the use of the present result, especially the hyperforce sum rule, in the context of machine learning interatomic potentials \cite{mace, nequip, equiformer, equiformer_v2}. It is also an interesting problem to come up with an idea to parametrize diffeomorphisms, to enable gradient-based optimization to accelerate molecular simulations.  

\appendix

\section{Definition and Basic Properties of Tempered Distributions} \label{app:distribution}

We recall the definition and some basic properties of tempered distributions used in the main text. Readers can refer to, for example, \cite{anal_pdo} for further detail.

\begin{definition}
We denote by $\symscr{S}(\symbb{R}^{d})$ the set of all $C^{\infty}$ functions $\phi$ such that 
\[
\sup_{x \in \symbb{R}^{d}} \left| \left(1 + \| x \|\right)^{k/2} D^{\alpha} \phi(x)\right| < \infty
\]
for all $k \in \symbb{N}_{0}$ and multi-indices $\alpha \in \symbb{N}^{d}_{0}$. We call $\symscr{S}(\symbb{R}^{d})$ the \textup{Schwartz space} or the set of \textup{Schwartz functions}. 
We often say that a function $\phi \in \symscr{S}(\symbb{R}^{d})$ is a rapidly decreasing function. 
\end{definition}
We can endow $\symscr{S}(\symbb{R}^{d})$ with a semi-norm topology based on the definition above. Specifically, the following family of semi-norms makes $\symscr{S}(\symbb{R}^{d})$ a Fr\'{e}chet space: 
\[
\left\{ p_{\alpha , k}(\phi) 
= \sup_{x \in \symbb{R}^{d}} \left| \left(1 + \| x \|\right)^{k/2} D^{\alpha} \phi(x)\right| 
\, ; \, \forall k \in \symbb{N}_{0},\,\forall \alpha \in \symbb{N}^{d}_{0}\right\}.
\] 
This semi-norms gives a definition of the convergence of sequences in the Schwartz space: A sequence $\{\phi_{n}\} \subset \symscr{S}(\symbb{R}^{d})$ converges to $\phi \in \symscr{S}(\symbb{R}^{d})$ in $\symscr{S}$ when $p_{\alpha , k}(\phi_{n} - \phi) \to 0$ for any $k \in \symbb{N}_{0}$ and $\alpha \in \symbb{N}^{d}_{0}$. This further gives the definition of the convergence for the space of tempered distribution.

\begin{definition} 
\label{def:tempered}
A continuous linear form $u : \symscr{S}(\symbb{R}^{d}) \rightarrow \symbb{C}$ is called a \textup{tempered distribution}. The set of all tempered distributions is denoted by $\symscr{S}^{\prime}(\symbb{R}^{d})$. 
A sequence 
$\{u_{j}\} \subset \symscr{S}^{\prime}(\symbb{R}^{d})$ 
is said to converge to $u \in \symscr{S}^{\prime}(\symbb{R}^{d})$ in $\symscr{S}'$ 
when $\displaystyle \lim_{j \rightarrow \infty} u_{j}(\phi) = u(\phi)$ for every $\phi \in \symscr{S}(\symbb{R}^{d})$. 
\end{definition}

The convergence at each point $\phi \in \symscr{S}(\symbb{R}^{d})$ induces the convergence in $\symscr{S}'$. 

\begin{theorem}
Let 
$\{u_{j}\} \subset \symscr{S}^{\prime}(\symbb{R}^{d})$ 
be a sequence in $\symscr{S}^{\prime}(\symbb{R}^{d})$. 
If a sequence $\{ u_{j}(\phi) \}$ in $\symbb{C}$ converges for any $\phi \in \symscr{S}(\symbb{R}^{d})$, 
there exists a tempered distribution $u \in \symscr{S}^{\prime}(\symbb{R}^{d})$ such that $u_{j}$ converges to $u$ in $\symscr{S}'$. 
\end{theorem}

It is also conventional and practical to adopt a pairing notation for the distribution. 
Specifically, it is typically written as 
\[ 
\langle u, \phi \rangle = u(\phi) 
\]
for $u \in \symscr{S}^{\prime}(\symbb{R}^{d})$ and 
$\phi \in \symscr{S}(\symbb{R}^{d})$. 
This definition allows us to identify continuous functions on $\symbb{R}^{d}$ with tempered distributions by assigning integration. 

\begin{example}
\label{def:tempered_function}
Let $f$ be a continuous function on $\symbb{R}^{d}$. 
Assume that $f$ is \textup{tempered}, 
that is, there exists 
a constant $C > 0$ and $\ell > 0$ such that 
\begin{equation} 
|f(x)| \leq C (1 + \| x \|)^{\ell}  
\end{equation} 
for all $x \in \symbb{R}^{d}$. 
Then 
the following functional 
\[ 
u_{f} \colon 
\symscr{S}(\symbb{R}^{d}) \rightarrow \symbb{C}, 
\qquad \phi \mapsto \langle u_{f} , \phi \rangle = \int_{\symbb{R}^{d}} f(x) \phi(x) \, dx   
\] 
is well-defined and a tempered distribution. 
This identification is legitimate based on a fact that $u_{f} =u_{g}$ indicates $f = g$. 
\end{example}

The second important example is the ``delta functional''. 
\begin{example}
For any $a \in \symbb{R}^{d}$, the delta functional $\delta_{a}$ at $a$ is defined as follows: 
\begin{equation}
\langle \delta_{a} , \phi \rangle = \phi (a), 
\quad \text{for } \phi \in \symscr{S}(\symbb{R}^{d}). 
\end{equation} 
This is a conceptual counter part of ``delta function'', written by $\delta (x - a)$, in physics literature. 
See also \cite[Section 5]{MR4412551} for details of delta functional. 
\end{example}

We introduce the notion of differentiation and multiplication for distributions. 

\begin{definition} 
For $u \in \symscr{S}^{\prime}(\symbb{R}^{d})$, we set
\[ 
\langle \partial_{k} u, \phi \rangle = -\langle u, \partial_{k} \phi\rangle, 
\quad \text{for }  \phi \in \symscr{S}(\symbb{R}^{d}). 
\]
For a  tempered $C^{\infty}$ function $f$,  
we define
\[ 
\langle f u, \phi \rangle = \langle u, f\phi\rangle, 
\quad \text{for } \phi \in \symscr{S}(\symbb{R}^{d}). 
\]
\end{definition} 

\begin{example} 
The above definitions are originated on some properties of the integral-type distributions. 
Indeed, if we have a  tempered $C^{1}$ function $u$ on $\symbb{R}^{d}$,  
we obtain by the integration by parts 
\[ 
\int_{\symbb{R}^{d}} (\partial_{k} u)(x) \phi (x) \, dx 
= - \int_{\symbb{R}^{d}} u(x) (\partial_{k} \phi)(x) \, dx, \quad \text{for } \phi \in \symscr{S}(\symbb{R}^{d}).
\]
If $f$ is a  tempered $C^{\infty}$ function, 
then we also get 
\[
\int_{\symbb{R}^{d}} (f u)(x) \phi(x) \, dx 
= \int_{\symbb{R}^{d}} u(x) (f \phi)(x) \, dx, 
\quad \text{for } \phi \in \symscr{S}(\symbb{R}^{d}),  
\]
in which $f\phi$ again belongs to $\symscr{S}(\symbb{R}^{d})$. 
\end{example}

For a diffeomorphism $f \colon \symbb{R}^{d} \to \symbb{R}^{d}$ and 
a rapidly decreasing function $\phi \in \mathcal{S}(\symbb{R}^{d})$, 
the pullback function $f^{\ast}\phi  = \phi \circ f$ 
is not always rapidly decreasing. Therefore, the pullback with a \textit{general} diffeomorphism does not define a well-defined map between the Schwartz spaces.
However, the pullback is well-defined when limiting to certain classes of diffeomorphisms. Such a class includes diffeomorphisms with compact supports. 
\begin{definition} 
\label{def:composition}    
Let $g \colon \symbb{R}^{d} \to \symbb{R}^{d}$ be a diffeomorphism between $\symbb{R}^{d}$ such that there exists 
a compact subset $K \subset \symbb{R}^{d}$ 
and an invertible matrix $T$ of order $d$ 
such that $g|_{\symbb{R}^{d} \setminus K} = T$. 
Then, noting that the pullback function 
$g^{\ast}\phi = \phi \circ g^{-1}$ 
for $\phi \in \symscr{S}(\symbb{R}^{d})$ 
is rapidly decreasing by the multivariate Faà di Bruno's formula (see Appendix \ref{sec:multi_faddibruno}), we have a continuous map $g^{*}: \symscr{S}^{\prime}(\symbb{R}^{d}) \rightarrow \symscr{S}^{\prime}(\symbb{R}^{d})$ defined as follows:
\[ 
\langle g^{*}u, \phi \rangle = \langle u , \phi \circ g^{-1} \rangle 
\quad 
\text{for } u \in \symscr{S}^{\prime}(\symbb{R}^{d}) \text{ and } \phi \in \symscr{S}(\symbb{R}^{d}). 
\]
We call $g^{*}u$ the \textup{composition} or \textup{pullback} of $u$ by $g$. 
\end{definition} 

When the distribution is of the integral type with a function, the pullback of the distribution can be interpreted through the change of variables: 
\begin{example} 
\label{exm:pullback_function}
Let $f$ be a tempered continuous function on $\symbb{R}^{d}$ and 
$g$ a diffeomorphism between $\symbb{R}^{d}$ such that there exists 
a compact subset $K \subset \symbb{R}^{d}$ 
and an invertible matrix $T$ of order $d$ 
such that $g|_{\symbb{R}^{d} \setminus K} = T$. 
Then the composition $g^{\ast}u_{f}$ is 
given by $g^{\ast}u_{f} = u_{\tilde{g}(f)}$, where $\tilde{g}(f) = f \circ g \cdot |\det g'|$. 
\end{example} 
We can generalize the value of a distribution to a Banach space. 

\begin{remark}
Let $B$ be a Banach space. 
By generalizing Definition \ref{def:tempered}, 
a continuous linear map $u \colon \symscr{S}(\symbb{R}^{d}) \to B$ is called a $B$-valued tempered distribution.  
The set of all tempered distributions is denoted by $\symscr{S}^{\prime}(\symbb{R}^{d} ; B)$. 
By reinterpreting the absolute value of a complex number to the norm of $B$, we can define convergence, differentiation, multiplication and composition of $B$-valued tempered distributions in a manner similar to the arguments up until Definition \ref{def:composition}. 
\end{remark}

Here, we also recall that the product rule for derivative (Leibniz type formula) holds for tempered distributions and Schwartz functions. 
We employ the uniform boundedness principle of Fr\'{e}chet spaces  \cite{zbMATH03320477}. 

\begin{theorem} 
\label{thm:uniform_boundedness}
Let $X \subset \symscr{S}^{\prime}(\symbb{R}^{d} ; B)$ 
be a bounded set, that is, for any $\phi \in \symscr{S}(\symbb{R}^{d})$, the set $\{ \langle u , \phi \rangle \in B \,;\, u \in X \}$  
is bounded in $B$. 
Then there exists a constant $C > 0$ and \textup{finite} family of seminorms $\{ p_{\alpha_{i} , k_{j}} \}$ such that 
the following formula holds 
for any $u \in B$ and $\phi \in \symscr{S}(\symbb{R}^{d})$: 
\begin{equation}
\| \langle u , \phi \rangle \|_{B} 
\leq C \max_{i,j} p_{\alpha_{i} , k_{j}}(\phi). 
\end{equation}
\end{theorem}

\begin{definition}
Let $\phi_{t} \in \symscr{S}(\symbb{R}^{d})$ $(t \in \symbb{R})$ 
be a family of rapidly decreasing functions. 
We call that $\phi_{t}$ is differentiable in $\symscr{S}$ at $t=a$ 
if the limit $\displaystyle \lim_{h \to 0} \frac{\phi_{a+h} - \phi_{a}}{h}$ 
converges in $\symscr{S}$. 
We denote the limit by $\displaystyle \diff{\phi_{t}}{t}(a)$. 
\end{definition}

\begin{definition}
Let $u_{t} \in \symscr{S}'(\symbb{R}^{d}; B)$ $(t \in \symbb{R})$ 
be a family of $B$-valued tempered distributions. 
We call that $u_{t}$ is differentiable in $\symscr{S}^{\prime}$ at $t=a$ 
if the limit $\displaystyle \lim_{h \to 0} \frac{u_{a+h} - u_{a}}{h}$ 
converges in $\symscr{S}'$. 
We denote the limit by $\displaystyle \diff{u_{t}}{t}(a)$. 
\end{definition}

\begin{theorem} 
\label{thm:product_rule}
Let $\{\phi_{t}\}_{t \in \symbb{R}} \subset \symscr{S}\left(\symbb{R}^{d}\right)$ be a differentiable family in $\symscr{S}$ and 
$\{u_{t}\}_{t \in \symbb{R}} \subset \symscr{S}^{\prime}\left(\symbb{R}^{d} ; B\right)$ be a differentiable family in $\symscr{S}^{\prime}$. 
Then we have 
\begin{equation} 
\diff{}{t} \left\langle u_{t} , \phi_{t} \right\rangle 
= \left\langle \diff{u_{t}}{t} , \phi_{t} \right\rangle 
    + \left\langle u_{t} , \diff{\phi_{t}}{t} \right\rangle. 
\label{eq:product_rule}
\end{equation}
\end{theorem}

\begin{proof} 
By the definition of derivative, we have 
\begin{align}
\diff{}{t} \langle u_{t} , \phi_{t} \rangle  &= \lim_{h \rightarrow 0} \frac{\langle u_{t+h} , \phi_{t+h} \rangle  - \langle u_{t} , \phi_{t} \rangle }{h} \\
&= \lim_{h \rightarrow 0} \left\langle \frac{u_{t+h} - u_{t}}{h}, \phi_{t} \right\rangle  + \lim_{h \rightarrow 0} \left\langle u_{t+h}, \frac{ \phi_{t+h} - \phi_{t}}{h} \right\rangle. 
\end{align}
The first term converges to 
$\displaystyle \left\langle \diff{u_{t}}{t} , \phi_{t} \right\rangle$ 
by the definition of the convergence sequence 
in $\symscr{S}^{\prime}\left(\symbb{R}^{d}\right)$. 
The convergence of the second term can be shown by using uniform boundedness principle; see Theorem \ref{thm:uniform_boundedness}. 
In the sequel in the proof, we will prove the formula 
$\displaystyle \lim_{h \rightarrow 0} \left\langle u_{t+h}, \frac{ \phi_{t+h} - \phi_{t}}{h} \right\rangle = \left\langle u_{t}, \diff{\phi_{t}}{t}  \right\rangle$. 
We first rewrite 
\begin{align}
&\phantom{=}\left\langle u_{t+h}, \frac{ \phi_{t+h} - \phi_{t}}{h} \right\rangle  - \left\langle u_{t}, \diff{\phi_{t}}{t}  \right\rangle  \\ 
&= \left\langle u_{t+h}, \frac{ \phi_{t+h} - \phi_{h}}{h} - \diff{\phi_{t}}{t} \right\rangle  + \left\langle u_{t+h} - u_{t}, \diff{\phi_{t}}{t} \right\rangle. 
\label{eq:expansion_uni_bound_thm}
\end{align} 
Since a family $\{ u_{t+h} \}_{h}$ is a convergent sequence in $\symscr{S}^{\prime}$, for a sufficiently small $\epsilon$, 
a subset $\{ \langle u_{t+h}, \phi \rangle \}_{h \in [-\epsilon , \epsilon]} \subset B$ is bounded for any $\phi \in \symscr{S}(\symbb{R}^{d})$. 
By Theorem \ref{thm:uniform_boundedness}, 
there exists a constant $C > 0$ and a finite family of seminorms $\{ p_{\alpha_{i} , k_{j}} \}$ such that 
\[
\left\| \left\langle u_{t+h} , \frac{\phi_{t+h} - \phi_{t}}{h}  - \diff{\phi_{t}}{t} \right\rangle \right\|_{B} 
\leq C \cdot \max_{i,j} \symbfit{p}_{\alpha_{i} , k_{j}}\left(\frac{\phi_{t+h} - \phi_{t}}{h}  - \diff{\phi_{t}}{t} \right) \xrightarrow{\ h \rightarrow 0 \ } 0. 
\]  
The second term of \eqref{eq:expansion_uni_bound_thm} also converges to $0$ by the definition of the convergence sequence in $\symscr{S}^{\prime}(\symbb{R}^{d})$, which concludes the proof. 
\end{proof}

\section{Multivariate Faà di Bruno's formula} 
\label{sec:multi_faddibruno}
We refer to \cite{MR1325915}. Let $\symbb{N}_{0}$ be the set of nonnegative integers. For $\symbfit{\nu} = (\nu_{1}, \cdots, \nu_{\ell}) \in \symbb{N}_{0}^{\ell}$ and $\mathbf{z} \in \symbb{R}^{\ell}$, in which $\ell$ is a positive integer, we define 
\[|\symbfit{\nu}| = \sum_{i=1}^{\ell} \nu_{i}, \quad \mathbf{z}^{\symbfit{\nu}} = \prod_{i=1}^{\ell} z_{i}^{\nu_{i}}.\]
For a smooth function $h: \symbb{R}^{\ell} \rightarrow \symbb{R}$ defined by the composition of 
\[
\symbfit{g} = (g^{(1)}(x_{1}, \dots, x_{\ell}), \dots, g^{(m)}(x_{1}, \dots, x_{\ell})): \symbb{R}^{\ell} \rightarrow \symbb{R}^{m}
\] 
and $f: \symbb{R}^{m} \rightarrow \symbb{R}$, we further define the following notations:
\[ 
g^{(i)}_{\symbfit{\mu}} = D_{\symbfit{x}}^{\symbfit{\nu}} g^{(i)}(\symbfit{x}_{0}), \quad \symbfit{g}_{\symbfit{\mu}} = (g^{(1)}_{\symbfit{\mu}}, \dots, g^{(\ell)}_{\symbfit{\mu}}), \quad h_{\symbfit{\nu}} = D_{\symbfit{x}}^{\symbfit{\nu}} h(\symbfit{x}_{0}), \quad f_{\symbfit{\lambda}} = D_{\symbfit{y}}^{\symbfit{\lambda}} f(\symbfit{y}_{0}), 
\]
in which $\symbfit{\lambda} \in \symbb{N}_{0}^{m}$, $\symbfit{\mu} \in \symbb{N}_{0}^{\ell}$, and $D^{\symbfit{\nu}}_{\symbfit{x}}$ is defined as
\[D^{\symbfit{\nu}}_{\symbfit{x}} = \frac{\partial^{|\symbfit{\nu}|}}{\partial x_{1}^{\nu_{1}} \cdots \partial x_{\ell}^{\nu_{\ell}}}, \quad \text{for } |\nu| > 0.\]
We introduce a linear order on $\symbb{N}^{\ell}_{0}$. For $\symbfit{\nu}, \symbfit{\mu} \in \symbb{N}^{\ell}_{0}$, we write $\symbfit{\mu} \prec \symbfit{\nu}$ if one of the following conditions holds:
\begin{itemize}
    \item[(i)] $|\symbfit{\nu}| < |\symbfit{\mu}|$;
    \item[(ii)] $|\symbfit{\nu}| = |\symbfit{\mu}|$ and $\mu_{1} < \nu_{1}$; or
    \item[(iii)] $|\symbfit{\nu}| = |\symbfit{\mu}|$, $\mu_{1} = \nu_{1}$, $\dots$, $\mu_{k} = \nu_{k}$ and $\mu_{k+1} < \nu_{k+1}$ for some $1 \leq k < \ell$.
\end{itemize}
We also define a set of ordered multi-indices as follows: 
\[p_{s}(\symbfit{\nu}, \symbfit{\lambda}) = \{(\symbfit{k}_{1}, \dots, \symbfit{k}_{s}, \symbfit{\ell}_{1}, \dots, \symbfit{\ell}_{s}): |\symbfit{k}_{i}|>0,\ \symbfit{0} \prec \symbfit{\ell}_{1} \prec \cdots \prec \symbfit{\ell}_{s},\ \sum_{i=1}^{s}\symbfit{k}_{i} = \symbfit{\lambda}, \ \sum_{i=1}^{s}|\symbfit{k}_{i}| \symbfit{\ell}_{i} = \symbfit{\nu}\}.\]

Then, a multivariate Faà di Bruno's formula is 
\begin{align}
h_{\symbfit{\nu}} = \sum_{1 \leq |\symbfit{\lambda}| \leq n} f_{\symbfit{\lambda}} \sum_{s=1}^{n} \sum_{p_{s}(\symbfit{\nu}, \symbfit{\lambda})}(\symbfit{\nu}!) \prod_{j=1}^{s} \frac{[\symbfit{g}_{\symbfit{\ell}_{j}}]^{\symbfit{k}_{j}}}{(\symbfit{k}_{j}!) [\symbfit{\ell}_{j}!]^{|\symbfit{k}_{j}|}}, \label{eq:multi_faadibruno}   
\end{align}
in which $n = |\symbfit{\nu}|$ and we set $0^{0} = 1$.

\section{(Partially) extended real numbers} \label{app:extended_real} We give the definition and basic properties of extended real numbers. Readers who need its more elaborated exposition can refer to, for example \cite{basic_real_analysis}. 

The extended real line $\bar{\symbb{R}}$ is the disjoint union $\bar{\symbb{R}} = \{-\infty\} \cup \symbb{R} \cup \{+\infty\}$, with the following properties:
\begin{itemize}
    \item[1.] $x \in \symbb{R} \Rightarrow -\infty < x < +\infty$;
    \item[2.] $x \in \symbb{R} \Rightarrow x + \infty = + \infty, \ x - (+\infty) = - \infty, \ x/(\pm \infty) = 0$;
    \item[3.] $x > 0 \Rightarrow x \cdot (+ \infty) = +\infty, \ x \cdot (- \infty) = -\infty$;
    \item[4.] $x < 0 \Rightarrow x \cdot (+ \infty) = -\infty, \ x \cdot (- \infty) = +\infty$;
    \item[5.] $(+ \infty) + (+\infty) = +\infty, \ (- \infty) - (+\infty) = -\infty$;\ and 
    \item[6.] $(+ \infty) \cdot (\pm\infty) = \pm \infty, \ (- \infty) \cdot (\pm\infty) = \mp\infty$. 
\end{itemize}
We extend the inequalities defined on $\symbb{R}$ to $\bar{\symbb{R}}$ by defining $c < +\infty$ and $-\infty < c$ for any real numbers $c \in \symbb{R}$. We further define 
\begin{align}
    [c, +\infty] &= [c, +\infty) \cup \{+ \infty\} = \{p \in \bar{\symbb{R}} \, ; \, p \geq c \}, \\
    (c, +\infty] &= (c, +\infty) \cup \{+ \infty\} = \{p \in \bar{\symbb{R}} \, ; \, p > c \}, \\
    [-\infty, c] &= (-\infty, c] \cup \{- \infty\} = \{p \in \bar{\symbb{R}} \, ; \, p \leq c \}, \\
    [-\infty, c) &= (-\infty, c) \cup \{- \infty\} = \{p \in \bar{\symbb{R}} \, ; \, p < c \}. 
\end{align}
There is an order-preserving bijective function between $\bar{\symbb{R}}$ and $[-1, 1]$:
\[\phi(c) = \frac{c}{1 + |c|}.\]
If we define $d(c, c^{\prime}) = |\phi(c) - \phi(c^{\prime})|$, this gives a metric on $\bar{\symbb{R}}$. Therefore, $(\bar{\symbb{R}}, d)$ is a metric space and the space is compact because it is homeomorphic to $[-1, 1]$ through the bijection $\phi$. 

We define the continuity of functions defined on $\bar{\symbb{R}}$. The definition is a straightforward extension of the continuity notion for functions on $\symbb{R}$. This faithful extension is possible because the convergence of an infinite sequence in $\bar{\symbb{R}}$ translates to the convergence of the sequence mapped by $\phi$ to $[-1, 1]$.
\begin{definition}
A function $f: \bar{\symbb{R}} \rightarrow \bar{\symbb{R}}$ is continuous if any convergent sequence $\{c_{n}\} \subset \bar{\symbb{R}}$, we have 
\[
\lim_{n \rightarrow \infty} f(c_{n}) = f \left( \lim_{n \rightarrow \infty} c_{n} \right).
\]
\end{definition}

In general, it is non-trivial to give a formal definition of the derivative for $\bar{\symbb{R}}$-valued functions because $\pm\infty - (\pm \infty)$ is not defined. However, this situation may be defused when  infinite values can be ``hidden" by applying another function altering the infinite numbers to zero. 

This example includes Hamiltonians $\symcal{H}$ defined in Definition \ref{def:real_hamiltonian}. A Hamiltonian $H \in \symcal{H}$ is defined as a $\symbb{R} \cup \{ \infty \}$-valued continuous function. When we apply the exponential map $e^{-x}$ to $H$, then $e^{-H(r^N, p^N)}$ turns out to be a function from $\symbb{R}^{dN \times 2}$ to $\symbb{R}$. This function is also continuous in the standard sense, because $\symbb{R} \cup \{ \infty \}$ is homeomorphic to (-1, 1]. 

\bibliographystyle{plain}
\bibliography{ref_168}

\end{document}